\documentclass[11pt]{amsart} 
\usepackage{amsmath, amssymb, latexsym, tikz}
\usepackage{caption,ifthen,url,fancyhdr,cite,longtable}
\usepackage[foot]{amsaddr}    
\usepackage{libertine,lineno,lscape} 
\usepackage[a4paper, left=2.3cm, right=2.3cm, top=3cm, bottom=3cm]{geometry}  
  
\usepackage{arydshln}

\usetikzlibrary{arrows,decorations.markings}
\usetikzlibrary{matrix}
\usetikzlibrary{graphs}
\usetikzlibrary{backgrounds}
   
\newtheoremstyle{mythm}                   
{6pt}
{6pt}
{\it}
{}
{\bf}
{.}
{.5em}
{}

\newtheoremstyle{mydef}                   
{6pt}
{6pt}
{}
{}
{\bf}
{.}
{.5em}
{}

\newtheoremstyle{myrem}                   
{6pt}
{6pt}
{}
{}
{\bf}
{.}
{.5em}
{}

\theoremstyle{mythm}      
\newtheorem{theorem}{Theorem}[section]
\newtheorem{proposition}[theorem]{Proposition}
\newtheorem{lemma}[theorem]{Lemma}

\theoremstyle{mydef}

\theoremstyle{myrem}
\newtheorem{remark}[theorem]{Remark}

\numberwithin{equation}{section}

\numberwithin{equation}{section}

\newcommand{\C}{\mathbb{C}}
\newcommand{\Q}{\mathbb{Q}}
\newcommand{\Z}{\mathbb{Z}}
\newcommand{\R}{\mathbb{R}}

\newcommand{\GL}{\mathrm{\mathop{GL}}}
\newcommand{\SL}{\mathrm{\mathop{SL}}}

\newcommand{\ad}{\mathrm{\mathop{ad}}}
\newcommand{\diag}{\mathrm{\mathop{diag}}}

\newcommand{\gl}{\mathfrak{\mathop{gl}}}

\newcommand{\ssl}{\mathfrak{\mathop{sl}}}

\newcommand{\so}{\mathfrak{\mathop{so}}}
\newcommand{\ssp}{\mathfrak{\mathop{sp}}}
\newcommand{\g}{\mathfrak{g}}
\newcommand{\h}{\mathfrak{h}}

\renewcommand{\u}{\mathfrak{u}}

\newcommand{\z}{\mathfrak{z}}
\renewcommand{\a}{\mathfrak{a}}
\newcommand{\kk}{\mathfrak{k}}
\newcommand{\pp}{\mathfrak{p}}

\newcommand{\uu}{\mathfrak{u}}

\newcommand{\wG}{\widehat{G}}

\newcommand{\Dih}{{\rm Dih}}

\newcommand{\SmallMatrix}[1]{\text{\begin{tiny}${\arraycolsep=0.4\arraycolsep\ensuremath
{\begin{pmatrix}#1\end{pmatrix}}}$\end{tiny}}}

\renewcommand{\leq}{\leqslant}

\renewcommand{\geq}{\geqslant}

\setpagewiselinenumbers
\reversemarginpar
\subjclass[2000]{}

\newcommand{\mye}[1]{\text{\small$|#1\rangle$}}

\newcounter{ithmcount}
\newenvironment{iprf}{\begin{list}{{\rm
	\alph{ithmcount})}}{\usecounter{ithmcount}\labelwidth-5pt
      \leftmargin0pt \topsep3pt \itemsep1pt \parsep2pt}}{\qedhere\end{list}}

\newenvironment{ithm}{\begin{list}{{\rm \alph{ithmcount})}}{\usecounter{ithmcount}\labelwidth18pt
      \leftmargin18pt \topsep3pt \itemsep1pt \parsep2pt}}{\end{list}}

\begin{document}


\def\DynkinNodeSize{1.5mm}
\def\DynkinArrowLength{1.5mm}
\tikzset{
dnode/.style={
circle,
inner sep=0pt,
minimum size=\DynkinNodeSize,
fill=white,
draw},
middlearrow/.style={
decoration={markings,
mark=at position 0.6 with
{\draw (0:0mm) -- +(+135:\DynkinArrowLength); \draw (0:0mm) -- +(-135:\DynkinArrowLength);},
},
postaction={decorate}
},
leftrightarrow/.style={
decoration={markings,
mark=at position 0.999 with
{
\draw (0:0mm) -- +(+135:\DynkinArrowLength); \draw (0:0mm) -- +(-135:\DynkinArrowLength);
},
mark=at position 0.001 with
{
\draw (0:0mm) -- +(+45:\DynkinArrowLength); \draw (0:0mm) -- +(-45:\DynkinArrowLength);
},
},
postaction={decorate}
},
sedge/.style={
},
dedge/.style={
middlearrow,
double distance=0.5mm,
},
tedge/.style={
middlearrow,
double distance=1.0mm+\pgflinewidth,
postaction={draw}, 
},
infedge/.style={
leftrightarrow,
double distance=0.5mm,
},
}


\title{Classification of four qubit states and their stabilisers under SLOCC operations}
\subjclass[2000]{}
\author[H. Dietrich]{Heiko Dietrich}
\address[Dietrich, Origlia]{School of Mathematics, Monash University, Clayton VIC 3800, Australia}
\author[W. A. de Graaf]{Willem A.\ de Graaf}
\address[de Graaf]{Department of Mathematics, University of Trento, Povo (Trento), Italy}
\author[A. Marrani]{Alessio Marrani}
\address[Marrani]{Centro Ricerche "Enrico Fermi", Roma, Italy}
\author[M. Origlia]{Marcos Origlia}
\email{\rm heiko.dietrich@monash.edu, jazzphyzz@gmail.com, marcos.origlia@monash.edu, willem.degraaf@unitn.it}
\thanks{The first, second, and fourth author were supported by an Australian Research Council grant, identifier DP190100317.}
\date{\today}

\begin{abstract}
We classify four qubit states under SLOCC operations, that is, we classify
the orbits of the group $\SL(2,\C)^4$ on
the Hilbert space $\mathcal{H}_4 = (\mathbb{C}^2)^{\otimes 4}$. We approach the
classification by realising this representation as a symmetric space of
maximal rank. We first describe general methods for classifying the orbits of
such a space. We then apply these methods to obtain the orbits in our special case, resulting in a  complete and irredundant classification of $\SL(2,\C)^4$-orbits on $\mathcal{H}_4$. It follows that an element of $(\mathbb{C}^2)^{\otimes 4}$ is
conjugate to an element of precisely 87 classes of elements. Each of these classes
either consists of one element or of a parametrised family of elements, and the elements in the same class all have equal stabiliser in $\SL(2,\C)^4$.  We also present a complete and irredundant classification of elements and stabilisers up to the action of ${\rm Sym}_4\ltimes\SL(2,\C)^4$ where ${\rm Sym}_4$ permutes the four tensor factors of
$(\mathbb{C}^2)^{\otimes 4}$. 
\end{abstract}
 
\maketitle

\section{Introduction}

\noindent Entanglement is a fundamental notion in Quantum Information Theory (QIT).
The beginning of the $\text{XXI}^{\text{st}}$ century has witnessed many efforts
and advances in understanding the nature of entanglement (see the review
papers \cite{ES,HHHH09}). Since entangled states lie at the core of
quantum-enhanced applications it is crucially important to know which of
these states are \textit{equivalent}, in the sense that they are capable of
performing the same QIT tasks almost equally well. Therefore, 
the classification of the entanglement of pure multipartite quantum
states under the group of reversible Stochastic Local Quantum Operations
assisted by Classical Communication (SLOCC) is nowadays one of the most
prominent challenges in QIT~(see~\cite{HHHH09}).

Since entanglement is deeply related to the non-local properties of a state,
its intrinsic nature cannot be affected by local quantum operations,
implemented by the SLOCC group \cite{Ben,Dur}, which provides the most
general local operations that can be implemented without deteriorating the
quantum correlations shared by spatially separated physical systems. As
mentioned above, two states belonging to the same entanglement class would
be able to perform the same tasks, because one should be obtained with
nonzero probability from the other  using local invertible operations.
Group theoretically, SLOCC equivalence classes on $n$-qubit states are $\SL(2,\C)^n$-orbits in the space $\mathcal{H}_n = \C^2\otimes \cdots \otimes \C^2$
($n$ factors $\C^2$).

SLOCC classifications for $n=2$ and $n=3$ are easily determined, yielding 
two and six SLOCC orbits for 2- and 3-qubit states, respectively. In
particular, for what concerns the case of an entangled pure state of two
qubits ($n=2$), it is well-known that it can be converted to the singlet
state by SLOCC operations \cite{Lo-Popescu}. For what concerns three
entangled qubits ($n=3$), it was proved in a series of works \cite%
{Dur,Dur2,verstraete} that any state can be converted by SLOCC operations either to
the GHZ-state $\frac{1}{\sqrt{2}}\left( \left\vert 000\right\rangle
+\left\vert 111\right\rangle \right) $, or to the W-state $\frac{1}{\sqrt{3}}%
\left( \left\vert 001\right\rangle +\left\vert 010\right\rangle +\left\vert
001\right\rangle \right) $, thus yielding to two inequivalent ways of
entangling three qubits. In general, the GHZ
(Greenberger-Horne-Zeilinger)-state is considered as the state with the
genuine tripartite entanglement, whereas the W-state enjoys the peculiar
property of having the maximal expected amount of twopartite entanglement if
one party is traced out \cite{Dur}.

For $n$ qubits with $n\geq 4$ uncountably many SLOCC
classes arise \cite{Dur}. The case of four qubits ($n=4$) has been the
subject of a number of studies; without claim to completeness we mention
\cite{verstraete,djok,Cao-Wang,nilp,Buniy,Chen-Grassl,GA16,Lamata:2006b,Li:2007c,Wallach:2013}. Here we cannot review all these publications; we just mention the
following.
\begin{iprf}
\item[$\bullet$] Verstraete et al.\ \cite{verstraete} considered the orbits of $\SL(2,\C)^4$ on
  $\mathcal{H}_4$ and also allowed for permutations of the qubits, that is,
  they considered the action of $\mathcal{S} =\mathrm{Sym}_4\ltimes \SL(2,\C)^4
  $. Their main result is a list of nine classes such that each
  $\mathcal{S}$-orbit has a point in exactly one of the classes;  it
  may happen that different elements of the same class are $\mathcal{S}$-conjugate. In \cite{djok} this classification was again derived and corrected.
\item[$\bullet$]  Wallach \cite{Wallach:2013} also considered the 
  $\SL(2,\C)^4$-orbits in $\mathcal{H}_4$. His methods are based on 
  results of Kostant\ \& Rallis \cite{kora}. One of the main results is the
  statement that there are 90 ``types'' of orbits.
\end{iprf}

In this paper we present a classification of the orbits of $\SL(2,\C)^4$ on
$\mathcal{H}_4$. The method that we use is similar to the one employed by
Wallach \cite{Wallach:2013}. However, we also use concepts and methods
introduced by Vinberg \cite{vinberg}, and employ a similar scheme as used
by Vinberg\ \& Elashvili \cite{VE78}. The main idea is to realise the
representation of $\SL(2,\C)^4$ on $\mathcal{H}_4$ using a symmetric pair
of maximal rank corresponding to the simple Lie algebra of type D$_4$.
This yields a Jordan decomposition of the elements of $\mathcal{H}_4$, and
allows partitioning its elements and $\SL(2,\C)^4$-orbits into three
classes: semisimple, nilpotent and mixed. The nilpotent orbits can be classified
by general methods such as the ones described in \cite{vinberg79,graaf}. The semisimple orbits are
classified by exhibiting a Cartan subspace and studying the action of the
(finite) Weyl group on this space. The mixed orbits are classified by
listing the nilpotent orbits in the centraliser of a semisimple element.
More specifically we have the following:
\begin{iprf}
\item[$\bullet$] There are 31 nilpotent orbits with  representatives given in Table \ref{tabNP}. This has been proved in \cite{nilp} 
  by the Kostant-Sekiguchi correspondence; it  can also be derived using
  Vinberg's method of carrier algebras \cite{vinberg79}. In the
  remainder of this paper we will therefore not discuss the nilpotent case~further.
\item[$\bullet$] There are 10 parametrised classes of nonzero semisimple elements, as given
  in Table \ref{tabW}. Every semisimple element is $\SL(2,\C)^4$-conjugate to an element in precisely one of these classes. For each
  class we explicitly determine  a finite group $\Gamma$ with the property that
  two elements in the class are $\SL(2,\C)^4$-conjugate if and only if
  they are $\Gamma$-conjugate. Elements of different classes are not
  $\SL(2,\C)^4$-conjugate. Furthermore, the elements of a class all have the
  {\em same} stabiliser in $\SL(2,\C)^4$.
\item[$\bullet$] For each semisimple class we explicitly list representatives of the
  orbits of mixed type whose semisimple part comes from the given class, see Theorem \ref{thmME}.
  This amounts to listing the possible nilpotent parts up to the action of
  the centraliser of the semisimple part.
\end{iprf}

This yields the following theorem.

\begin{theorem}\label{thmMain}
There are 87 classes of elements of $\mathcal{H}_4$: 31 classes consist
of a single nilpotent element, 10 classes consist of semisimple elements, and
46  classes consist of mixed elements. Each element of $\mathcal{H}_4$
is $\SL(2,\C)^4$-conjugate to an element of precisely one class; elements of a class all have the same stabiliser in~$\SL(2,\C)^4$. 
\end{theorem}

In particular, our results yield the first complete and irredundant classification of the $\SL(2,\C)^4$-orbits on $\mathcal{H}_4$; this follows from Theorem \ref{thmSE} (semisimple), Theorem \ref{thmME} (mixed), Table \ref{tabNP} (nilpotent), together with Remarks \ref{remSym} and~\ref{GammaConjugacy}.  We also compare our classifications with those of Verstraete et al.\ \cite{verstraete} and
Chterental \& Djokovi\v{c} \cite{djok}, and we present the first complete and irredundant classification of $({\rm Sym}_4\ltimes \SL(2,\C)^4)$-orbits in $\mathcal{H}_4$, together with their stabilisers. 

In \cite{GKW} it is argued that in many contexts it is important to
determine the stabiliser (also called \emph{group of local symmetries})
of a given element in $\mathcal{H}_n$.
As a corollary to Theorem \ref{thmMain}, it follows that the stabiliser of any element of
$\mathcal{H}_4$ is conjugate in $\SL(2,\C)^4$ to one of 87 stabilisers.
In \cite{GKW} the orbits of $\psi,\phi\in\mathcal{H}_4$ are defined to 
have the same {\em type} if the stabilisers of $\psi$ and $\phi$ in
$\SL(2,\C)^4$ are conjugate in that group. Hence we conclude that there
are at most 87 types of orbits.

When we consider the action of ${\rm Sym}_4\ltimes \SL(2,\C)^4$
then every element in $\mathcal{H}_4$ is conjugate to an element of
precisely one of 27 classes. As mentioned above we
explicitly determine the stabilisers in $\SL(2,\C)^4$
for the elements in our \enlargethispage{6ex}
classification of $({\rm Sym}_4\ltimes \SL(2,\C)^4)$-orbits, see Tables \ref{tabZ}, \ref{tabMT}, and~\ref{tabNilZ}. From Table \ref{tabZ}, Row~1, it is seen that the stabiliser of a
generic element is a finite group of order 32. For $n$-qubits with $n\geq 5$  the
situation is completely different, as in those cases the stabiliser of a generic
element is trivial \cite{GKW}.


\vspace*{-1ex}

\subsection{Structure of this paper}
In Section \ref{secGT} we describe our general approach to classifying the orbits in
symmetric spaces of maximal rank. These are certain representations of
reductive algebraic groups that arise from a $\Z/2\Z$-grading of a semisimple
Lie algebra. In Section \ref{secOurSS} we show how the representation of
$\SL(2,\C)^4$ arises in this way. We then apply our methods to derive
a classification of the semisimple orbits, see Theorem \ref{thmSE}. The orbits
of mixed elements are determined in Section~\ref{secMixed}, see
Theorem \ref{thmME}. In Section~\ref{secDjok} we compare our classifications
with those of Verstraete et al.\ \cite{verstraete} and
Chterental \& Djokovi\v{c} \cite{djok}, and we present a complete
classification of $({\rm Sym}_4\ltimes \SL(2,\C)^4)$-orbits in $\mathcal{H}_4$. In Section \ref{secApps} we investigate the ring of invariants for our main
example. All explicit calculations have been done in GAP \cite{gap} using the GAP
packages SLA and Singular; the latter provides an interface to  the algebra
software Singular \cite{singular}.


\section{Orbits in symmetric spaces of maximal rank}\label{secGT}

\noindent We let $\g$ be a semisimple Lie algebra over $\C$ and suppose that 
$\g$ is furnished with a $\Z/2\Z$-grading, that is, $\g= \g_0\oplus \g_1$ with $[\g_i,\g_j]\leq \g_{i+j\bmod 2}$ for all $i,j$; in particular, $\g_0$ is a
subalgebra that acts on $\g_1$. It is also said that $(\g,\g_0)$ is a
{\em symmetric pair}. Associated with this grading is an automorphism
$\theta \colon \g\to \g$ of order~$2$ such that each $\g_i$ is the
$(-1)^i$-eigenspace of $\theta$.

Let $G$ be the adjoint group of $\g$, that is,
the identity component of the automorphism group of $\g$; the Lie algebra of
$G$ is $\ad_\g \g\cong \g$. Let $G_0$ be the connected algebraic subgroup of
$G$ with Lie algebra $\ad_\g \g_0$; note that  $G_0\leq G^\theta
=\{g\in G: \theta g=g\theta\}.$ The group $G_0$ acts naturally on
$\g_1$ and we are interested in listing the orbits of $G_0$ in $\g_1$; the study of these orbits was initiated  by Kostant \&  Rallis \cite{kora}. The group $G_0$ with its action on $\g_1$ is a special case of a $\theta$-group, a concept
introduced by Vinberg \cite{vinberg,vinberg79}, who also studied the orbits of $G_0$ on $\g_1$. Part of Vinberg's theory is covered by a recent book by Wallach \cite{wallach}; in the sequel we will mainly refer
to this book, although all cited results can also be found in the papers by Kostant \& Rallis and Vinberg.  

A {\em Cartan subspace} of the pair  $(\g,\g_0)$ is a subspace
of $\g_1$ maximal with respect to the property that its elements are commuting semisimple elements. By
\cite[Corollary 3.55]{wallach} any two Cartan subspaces are $G_0$-conjugate.
In particular, they have the same dimension, which is called the {\em rank} of
$(\g,\g_0)$. Here we assume that a Cartan subspace of $\g_1$ is also a Cartan
subalgebra of $\g$, that is, we assume that the rank of $(\g,\g_0)$
is equal to the rank of the root system of $\g$. Up to conjugacy, there exists
a unique symmetric pair $(\g,\g_0)$ of maximal rank, which can be constructed
as follows: The split real form $\g_\R$  with complexification
$\g$ has a Cartan decomposition $\g_\R = \kk_\R\oplus \pp_\R$, and letting $\g_0$ and $\g_1$ be the complexifications of $\kk_\R$ and $\pp_\R$, respectively, we obtain the symmetric pair of maximal rank. In Table \ref{tabMaxSp} we
list the symmetric spaces of maximal rank corresponding to the simple
complex Lie algebras.

\begin{table}[ht]\renewcommand\arraystretch{1.2}
\scalebox{0.88}{\begin{tabular}{r|c|c|c|c}
{\bf type} & ${\pmb{\g}}$ & $\pmb{\g_0}$ & $\pmb{\g_1}$ & {\bf degrees} \\
\hline
${\rm A}_{n-1}$ & $\ssl(n,\C)$ & $\so(n,\C)$ & $S_0^2{\bf n}$ & $2,3,4,\ldots,n$\\
${\rm B}_n$ & $\so(2n+1,\C)$ & $\so(n,\C)\oplus\so(n+1,\C)$ &
${\bf n}\otimes ({\bf n+1})$ & $2,4,6,\ldots,2n$\\
${\rm C}_n$ & $\ssp(2n,\C)$ & $\ssl(n,\C)$ & $S^2 {\bf n} \oplus \overline{S^2{\bf n}}$
& $2,4,6,\ldots,2n$\\
${\rm D}_n$ & $\so(2n,\C)$ & $\so(n,\C)\oplus \so(n,\C)$ & ${\bf n}\otimes {\bf n}$
& $n,2,4,6,\ldots,2n-2$\\
${\rm E}_6$ & ${\rm E}_6(\C)$ & $\ssp(8,\C)$ & $\wedge_0^4 {\bf 8}$ & $2,5,6,8,9,12$\\
${\rm E}_7$ & ${\rm E}_7(\C)$ & $\ssl(8,\C)$ & $\wedge^4 {\bf 8}$ & $2,6,8,10,12,14,18$\\
${\rm E}_8$ & ${\rm E}_8(\C)$ & $\so(16,\C)$ & $\underset{\text{semispinor}}{%
  \mathbf{128}}$ & $2,8,12,14,18,20,24,30$\\
${\rm F}_4$ & ${\rm F}_4(\C)$ & $\ssp(6,\C)\oplus \ssl(2,\C)$ & $(\wedge_0^3{\bf 6})\otimes
{\bf 2}$ & $2,6,8,12$\\
${\rm G}_2$ & ${\rm G}_2(\C)$ & $\ssl(2,\C)\oplus\ssl(2,\C)$ & $S^3{\bf 2}\otimes{\bf 2}$ &
$2,6$\\
\end{tabular}}\\[1ex]\caption{Symmetric spaces of maximal rank in simple Lie algebras defined over $\C$. The fourth column displays the
  structure of $\g_1$ as $\g_0$-module; here we denote an irreducible module
  by its dimension, and a notation like $\wedge_0^4 {\bf 8}$ indicates
  the quotient of $\wedge^4 {\bf 8}$ by the trivial 1-dimensional module. 
  The last column has the degrees of the homogeneous invariant polynomials that generate the invariant ring, cf. \cite[Table 1, p.~59]{humcox}}\label{tabMaxSp}.
\end{table}

Recall that  $x\in \g$ is semisimple (nilpotent) if the adjoint map $\ad x\colon \g\to \g$ is a semisimple (nilpotent) endomorphism. The study of the $G_0$-orbits in $\g_1$ starts with the following well-known lemma on the Jordan decomposition;  we refer to \cite[Proposition 3]{kora} for a proof.

\begin{lemma}\label{lem:jd}
If $x\in\g_1$, then $x=s+n$ for unique semisimple $s\in \g_1$ and nilpotent
$n\in \g_1$ with  $[n,s]=0$. 
\end{lemma}

Accordingly, the $G_0$-orbits in $\g_1$ split into three classes: the nilpotent
orbits (that consist entirely of nilpotent elements), the semisimple orbits
(that consist of semisimple elements) and the mixed orbits (consisting of
elements that are neither semisimple nor nilpotent). Methods for listing the
nilpotent orbits have been developed by Vinberg \cite{vinberg79} and de Graaf \cite{graaf11},  so we will not comment on that here. Instead, we will describe
how to list the semisimple and mixed orbits. Our methods for that are based on (and very similar to) methods for an analogous problem developed by Vinberg \& {\`E}la{\v{s}}vili \cite{VE78}. We note that in \cite{antonyan} and \cite{antelash} the authors
study the orbits in the symmetric spaces of maximal rank of types $E_7$ and $E_8$, respectively, using methods that are also based on the approach in  \cite{VE78}. 

We continue with a subsection that contains some results on Weyl groups. Subsequently, we discuss  semisimple orbits and mixed orbits. Throughout, we use standard notation for Lie algebras and their related combinatorial data (Cartan subalgebras, root systems, Weyl groups, etc), and
we refer to the books of Erdmann \& Wildon \cite{erdmann} or
Humphreys \cite{hum} for more details and background information.

\subsection{Root subsystems}\label{secWeyl}
\noindent Let $\g$ be a semisimple complex Lie algebra with Cartan subalgebra
$\h$ and corresponding root system $\Phi$ and Weyl group $W$. We write
$\g=\h\oplus\bigoplus_{\alpha\in\Phi}\g_\alpha$ for the corresponding root space
decomposition. A subset  $\Pi\subseteq \Phi$ is a \emph{root subsystem} if
for $\alpha,\beta\in \Pi$ we have $-\alpha\in \Pi$ and, if $\alpha+\beta
\in \Phi$, then $\alpha+\beta\in \Pi$. For $\alpha\in \Phi$ let $s_\alpha\in W$
be the corresponding reflection. The group $W$ acts on $\h$ by $s_\alpha(h) =
h-\alpha(h)h_\alpha$, where $h_\alpha$ is the unique element of $[\g_\alpha,\g_{-\alpha}]\leq \h$ with $\alpha(h_\alpha)=2$ (see \cite[Remark~2.9.9]{graaf}). If $w\in W$, $h\in \h$, and $\alpha\in\Phi$, then we sometimes abbreviate $wh=w(h)$ and $w\alpha=w(\alpha)$. We have
the following property
\begin{equation}\label{eq:weylact1}
  \alpha(s_\beta(h)) = s_\beta(\alpha)(h) \quad \text{ for all } \alpha,\beta\in\Phi
  \text{ and } h\in \h,
\end{equation}
which implies that
\begin{equation}\label{eq:weylact2}
w(\alpha)(h) = \alpha(w^{-1}(h))  \quad \text{ for all } \alpha\in\Phi, w\in W
  \text{ and } h\in \h.
\end{equation}
For $p\in \h$ we define $\Phi_p$ to be the annihilator of $p$ in $\Phi$, that is,\[\Phi_p=\{\alpha\in\Phi : \alpha(p)=0\}.\]
It is clear that $\Phi_p$ is a root subsystem of $\Phi$, but not all root
subsystems arise in this way. The next lemma gives a criterion to decide whether a root subsystem is of the form $\Phi_p$ for some $p\in\h$; recall that a root subsystem $\Psi\subseteq \Phi$
is {\em complete} if it is not properly contained in a root subsystem of
$\Phi$ of the same rank.

\begin{lemma}\label{lem:complete}
Let $\Psi\subseteq \Phi$ be a root subsystem. There exists $p\in \h$ with
$\Psi=\Phi_p$ if and only if $\Psi$ is complete.
\end{lemma}
 
\begin{proof}
We first show that a  root subsystem $\Pi\subseteq \Phi$ is complete if and
only if $V_\Pi\cap\Phi = \Pi$, where $V_\Pi$ is the $\Q$-space spanned by
$\Pi$: If  $\Pi$ is complete, then  $V_\Pi\cap \Phi=\Pi$ since $V_\Pi\cap\Phi$
is a root subsystem of $\Phi$ containing $\Pi$ of the same rank as $\Pi$. For
the converse suppose that $V_\Pi\cap \Phi=\Pi$. If $\Pi'\subseteq \Phi$ is a
root subsystem containing $\Pi$ and of the same rank as $\Pi$, then $V_{\Pi'}$
contains $V_\Pi$ and both spaces are of the same dimension, hence they are
equal. Thus, $\Pi'$ is contained in $V_{\Pi'}\cap \Phi = V_\Pi\cap \Phi = \Pi$,
so  $\Pi'=\Pi$. 
  
Suppose that $\Psi = \Phi_p$. If $\beta \in V_{\Psi} \cap\Phi$ then  $\beta$ is
a linear combination of elements of $\Phi_p$, hence $\beta(p)=0$ and
$\beta\in\Phi_p=\Psi$. It follows that $\Psi$ is complete. For the converse, suppose that $\Psi$ is complete. If $\Psi$ and $\Phi$ have equal rank,  then $\Psi=\Phi$ and $\Psi=\Phi_0$.
Now suppose that the rank $s$ of  $\Psi$ is  less then the rank of $\Phi$. 
Define $\uu = \{ p\in\h : \alpha(p)=0 \text{ for all } \alpha\in\Psi\}$
and $\uu^\circ = \{ p\in \uu : \beta(p)\neq 0 \text{ for all }
\beta\in\Phi\setminus\Psi\}$; note that $\dim \uu = \dim \h -s>0$.
If $\beta\in \Phi\setminus\Psi$, then $\beta$ is not contained in $V_{\Psi}$ by
our claim above. Thus, the space spanned by $\beta$ and $\Psi$ has dimension
$s+1$ and so $\{ u\in \uu : \beta(u) = 0\}$ has dimension $\dim \h -s-1$. 
This shows that $\beta$ is nonzero on $\uu$, hence the kernel of every
$\beta\in\Phi\setminus\Psi$ on $\uu$ has codimension $1$. Since any finite
union of codimension $1$ subspaces of $\uu$ does not cover $\uu$ it follows
that there is some $p\in\uu$ with $\beta(p)\ne 0$ for all
$\beta\in\Phi\setminus\Psi$, thus  $\uu^\circ\neq\emptyset$. Furthermore,
for any $p\in \uu^\circ$ we have $\Psi=\Phi_p$.
\end{proof}

For a root subsystem $\Psi\subseteq\Phi$ we define \[\h_\Psi^\circ = \{p\in \h : \Phi_p=\Psi\};\]
note that this is the set of all $p\in\h$ such that $\alpha(p)=0$ for every $\alpha\in \Psi$ and $\beta(p)\ne0$ for every $\beta\in\Phi\setminus\Psi$.

\begin{lemma}\label{lem:hpsi}
Let $\Psi,\Pi\subseteq \Phi$ be root subsystems and $u\in W$. Then $\Pi=u\Psi$
if and only if $\h_\Pi^\circ = u\h_\Psi^\circ$.  
\end{lemma}

\begin{proof}
First assume that $\Pi=u\Psi$, and note that   $u(\Phi\setminus \Psi)=\Phi\setminus\Pi$.
Let $h\in \h_\Psi^\circ$. If $\beta\in \Pi$, then 
$\beta = u\alpha$ for some $\alpha\in \Psi$, and \eqref{eq:weylact2} shows that
$\beta(uh) = (u\alpha)(uh) = \alpha(h)=0$. If $\beta\in \Phi\setminus\Pi$, then $\beta = u\alpha$ for
some $\alpha\in \Phi\setminus\Psi$, and \eqref{eq:weylact2}  yields $\beta(uh) = \alpha(h)\neq 0$. This shows that $uh\in \h_\Pi^\circ$.
Conversely, let $h\in \h_\Pi^\circ$. Since $\Psi = u^{-1}\Pi$, the previous argument shows that  $u^{-1} h\in \h_\Psi^\circ$, so  $h\in u\h_\Psi^\circ$. Thus,  $\h_\Pi^\circ = u\h_\Psi^\circ$, as claimed.

Now suppose that $\h_\Pi^\circ = u\h_\Psi^\circ$. The first part of the
proof shows that $u\h_\Psi^\circ= h_{u\Psi}^\circ$, so our assumption is that  $\h_\Pi^\circ= \h_{u\Psi}^\circ$. The definition of $\h_{\Pi}^\circ$ immediately implies that $u\Psi=\Pi$, as claimed.
\end{proof}

By definition, each $p\in \h$ lies in $\h_{\Phi_p}^\circ$. Thus, if $\Phi_1,\ldots,\Phi_r$ are, up to $W$-conjugacy, all the complete root subsystems of $\Phi$ (including $\emptyset$ and $\Phi$ itself), then every $W$-orbit in $\h$ has a point in a unique set $\h_{\Phi_i}^\circ$.

For a fixed complete root subsystem $\Psi\subseteq\Phi$ we now  characterise the $W$-conjugacy of elements in $\h_\Psi^\circ$. The proof of the next lemma follows well-known ideas (see, for example, \cite[\S 1.12]{humcox}), however, we could not find an exact reference in the literature.

\begin{lemma}\label{lem:stab} Let $\Delta = \{\alpha_1,\ldots,
  \alpha_\ell\}$ be a basis of simple roots of $\Phi$.
  \begin{ithm}
  \item Every $p\in \h$ is $W$-conjugate to an element in $C = \{ h\in \h :
    \alpha_i(h) \geq 0 \text{ for all } i\in\{1,\ldots,\ell\}\}$ where we
    write $z>0$ for a complex number  $z = x+\imath y$ with $x,y\in \R$ if
    either $x>0$, or $x=0$ and $y>0$. 
  \item If $p\in \h$, then the stabiliser $W_p=\{w\in W \colon w(p)=p\}$
    is  generated by $\{s_\alpha:\alpha\in \Phi_p\}$.
    \end{ithm}
\end{lemma}

\begin{proof}
Throughout the proof we abbreviate $h_i = h_{\alpha_i}$.     
  \begin{iprf}
  \item This is standard: We construct a sequence $k_1=p,k_2,k_3\ldots$ of
    $W$-conjugate elements until we find some $k_m\in C$. If $k_n$ is defined,
    but $k_n\notin C$, then $\alpha_i(k_n)<0$ for some $i$, and we set
    $k_{n+1}=s_i(k_n)=k_n+c_ih_i$ with $c_i=-\alpha_i(k_n)>0$. Thus, by
    construction, all elements in the sequence $k_1,k_2,\ldots$ are distinct;
    since the $W$-orbit of $p$ is finite, we will eventually construct an
    element $k_m\in C$.

\item We first show that if $p\in C$, then $W_p$ is generated by the $s_{\alpha_i}$
such that $\alpha_i(p)=0$. By definition of $C$, we have $\alpha(p)\geq 0$ for
every positive root and $\alpha(p)\leq 0$ for every negative one. Now let
$w=s_{i_1}\cdots s_{i_t}\in W_p$ be a reduced expression; the claim follows if
each $s_{i_k}\in W_p$. Write $p_{t+1}=p$ and $p_j = s_{i_j}\cdots s_{i_t}(p)$ for
$j\in\{1,\ldots,t\}$. It follows from \eqref{eq:weylact1} that
\[(\ast)\quad\alpha_{i_{j-1}}(p_j) = s_{i_t}\cdots s_{i_j}(\alpha_{i_{j-1}})(p)
\quad \text{ for all } j\in\{2,\ldots,t+1\}.\]
By \cite[Corollary~10.2]{hum}, each $s_{i_t}\ldots s_{i_{j-1}}
(\alpha_{i_{j-1}})$ is a negative root, so
$s_{i_t}\ldots s_{i_{j-1}}(\alpha_{i_{j-1}})(p)\leq 0$ for $p\in C$; since
$s_{i_{j-1}}(\alpha_{i_{j-1}})=-\alpha_{i_{j-1}}$, we deduce from $(\ast)$ that
$c_{j-1}=\alpha_{i_{j-1}}(p_j)$ satisfies $c_{j-1}\geq 0$. This shows that
$p_{j-1}=s_{i_{j-1}}(p_j) = p_j - c_{j-1}h_{i_{j-1}}$ with $c_{j-1}\geq 0$. Since this
holds for every $j$, the equality $w(p)=p$ implies that all $c_{j-1}=0$,
showing that each $p_j=p$, hence $\alpha_{i_{j-1}}(p)=0$, and so each $s_{i_{j-1}}\in W_p$.

Now let $p\in \h$ and choose $w\in W$ with $w(p)\in C$. The above shows that
$W_p=w^{-1}W_{w(p)}w$ is generated by  $\{s_{w^{-1}(\alpha)} :
\alpha\in \Phi_{w(p)}\}$; here we use  $w^{-1}s_\alpha w=s_{w^{-1}(\alpha)}$, see
\cite[Lemma 9.2]{hum}. On the other hand, $\Phi_{w(p)}=w(\Phi_p)$ by \eqref{eq:weylact2}, so $W_p$ is generated by $\{s_{w^{-1}(\alpha)} :\alpha\in w(\Phi_p)\}=\{s_\alpha:\alpha\in \Phi_p\}$.
\end{iprf}
\end{proof}

For a root subsystem $\Psi\subseteq \Phi$ we define \[W_\Psi=\langle s_\alpha : \alpha\in \Psi\rangle\quad\text{and}\quad \Gamma_\Psi = N_W(W_\Psi)/W_\Psi.\]

Let $\Psi\subseteq\Phi$ be a complete subsystem. The previous lemma shows that  the stabiliser  $W_p$ of  $p\in \h_\Psi^\circ$ in $W$ is generated by all $s_\alpha$ with $\alpha\in \Phi_p$. By definition, $\Phi_p=\Psi$, and so  $W_p=W_\Psi$. In particular, for every $w\in N_W(W_\Psi)$ we have $wW_pw^{-1} = W_p$, and  if $q=wp$, then $W_q=W_p$, or  $W_{\Phi_q} = W_{\Phi_p}$, which implies that $\Phi_q=\Phi_p=\Psi$, and so $q\in \h_\Psi^\circ$. We conclude that $\Gamma_\Psi$ acts naturally on $\h_\Psi^\circ$ and we have the following:
   
\begin{proposition}\label{prop:gampsi}
Let $\Psi\subseteq\Phi$ be a complete subsystem. Two elements $p,q\in \h_\Psi^\circ$ are $W$-conjugate if and only if they are $\Gamma_\Psi$-conjugate.
\end{proposition}
 
\begin{proof}
If $q=wp$ with $w\in W$, then $W_q=wW_pw^{-1}$; since $W_q=W_p=W_\Psi$, we have $w\in N_W(W_\Psi)$, so  $p$ and $q$ are $\Gamma_\Psi$-conjugate. The converse is obvious. 
\end{proof}


\subsection{Semisimple orbits}\label{secZ2}

\noindent We revert back to the set up from the start of the section, that is
we consider a symmetric pair $(\g,\g_0)$ of maximal rank. We are interested
in listing the semisimple $G_0$-orbits in $\g_1$. Recall that a Cartan subspace of $\g$ is a maximal subspace
of $\g_1$ consisting of commuting semisimple elements, and any two Cartan subspaces are $G_0$-conjugate, see 
\cite[Corollary 3.55]{wallach}. Thus, every semisimple $G_0$-orbit in $\g_1$ intersects any given
Cartan subspace nontrivially. For a Cartan subspace $\h\leq \g_1$ let
$$W_\h = N_{G_0}(\h)/Z_{G_0}(\h)$$
be the {\em little Weyl group}, also called the  {\em Weyl group of the
graded Lie algebra $\g$}. This group was studied in detail by Vinberg
\cite{vinberg}, who proved (among other things) that two elements of $\h$
are $G_0$-conjugate if and only if they are $W_\h$-conjugate, see \cite[Theorem 2]{vinberg} or \cite[Proposition 3.61]{wallach}.

From now on we fix a Cartan subspace $\h$ in $\g_1$. By the remarks above, the
classification of the semisimple $G_0$-orbits in $\g_1$ is reduced to the
classification of the $W_\h$-orbits in $\h$.  Let $\Phi$ be the root system of $\g$ with respect to $\h$. Let $W$ be the 
Weyl group of $\Phi$. By definition, $W_\h$ is naturally a subgroup of 
$N_G(\h)/Z_G(\h)$; the next result shows that we actually have equality.  This will allow us 
to identify $W$ and $W_\h$. 

\begin{lemma}
We have $W_\h=N_G(\h)/Z_G(\h)\cong W$.
\end{lemma}

\begin{proof}
  It is well-known that $W\cong N_G(\h)/Z_G(\h)$, see \cite[Lemma~5.2.22]{graaf}. To prove $W_\h=N_G(\h)/Z_G(\h)$ we fix $\alpha\in \Phi$ and show that  $W_\h$ contains an element that acts as  $s_\alpha$ on $\h$. Let $x\in \g_\alpha$ and  $h\in \h$;  applying $\theta$ (the automorphism of $\g$ defining the grading) to the equality
  $[h,x] = \alpha(h)x$, we see that $\theta(x)\in \g_{-\alpha}$. Let $\u(\alpha)$
  be the subalgebra of $\g$ generated by $\g_\alpha$ and $\g_{-\alpha}$, so $\u(\alpha)\cong \ssl(2,\C)$ is stable under $\theta$.   Let $U(\alpha)$ denote the connected subgroup of $G$ with Lie algebra
  $\ad_\g \u$, that is, $U(\alpha)$ is generated by $\exp (\ad_\g tx_\alpha)$ and
  $\exp (\ad_\g tx_{-\alpha})$ with $x_{\pm\alpha}\in \g_{\pm\alpha}$ and $t\in \C$.
  The automorphism group of $\ssl(2,\C)$ is the adjoint group $\mathrm{PSL}(2,\C)$. By
  general theory of semisimple algebraic groups, see \cite[p.\ 182]{graaf},
  there is a surjective morphism of algebraic groups $U(\alpha) \to  \mathrm{PSL}(2,\C)$. Let $g_\alpha\in U(\alpha)$ denote an inverse image under this morphism of the restriction $\theta|_{\u(\alpha)}\in\mathrm{PSL}(2,\C)$. If $g\in U(\alpha)$, then
  $g_\alpha g g_\alpha^{-1} = \theta g\theta^{-1}$:  this is easy to check for generators $\exp(\ad_\g t x_{\pm\alpha})$, the general case follows from that. If we take $g=g_\alpha$, then we get $\theta g_\alpha \theta^{-1} =g_\alpha$, which proves that    $g_\alpha\in G_0$. Note that  $\h = \langle h_\alpha \rangle \oplus
  \hat \h_\alpha$ where $\hat \h_\alpha = \{x\in \h : \alpha(x) = 0\}$, so all elements of $U(\alpha)$ act as the identity on $\hat \h_\alpha$.
  Furthermore, $h_\alpha\in \u(\alpha)$, so  $g_\alpha(h_\alpha) =
  \theta(h_\alpha) = -h_\alpha$. It follows that $g_\alpha$ acts as $s_\alpha$ on $\h$.
\end{proof}

Now our procedure for obtaining a classification of the semisimple
$G_0$-orbits in $\g_1$ is as follows. Dynkin devised an algorithm to
find the root subsystems of $\Phi$ up to $W$-conjugacy (see \cite{dyn}, and also
\cite[pp.~221]{graaf}). We use this algorithm to compute all subsystems
up to $W$-conjugacy and discard those
that are not complete. We also add the empty set to the list. Let
$\Pi_1,\ldots,\Pi_r$ denote the obtained subsystems. Then for each $\Pi_i$
we compute the subspace
$$\h_{\Pi_i} = \{ p\in \h : \alpha(p)=0 \text{ for all } \alpha\in\Pi_i\}.$$
We call the $r$ sets $\h_{\Pi_i}^\circ$ the \emph{canonical semisimple sets}. Note that $p\in \h_{\Pi_i}$ lies in $\h_{\Pi_i}^\circ$ if and only if
$\beta(p)\neq 0$ for all $\beta\in\Phi\setminus\Pi_i$; this leads to a finite number of linear conditions for $\beta$. Multiplying them, we obtain a polynomial
function $F_i$ on $\h_{\Pi_i}$ such that $p\in \h_{\Pi_i}$ lies in $\h_{\Pi_i}^\circ$
if and only if $F_i(p)\neq 0$. We determine the groups
$\Gamma_{\Pi_i}$ by computing the normalisers $N_W(W_{\Pi_i})$. (In our main example discussed below we construct the quotient $\Gamma_{\Pi_i}$ as a complement to $W_{\Pi_i}$ in $N_W(W_{\Pi_i})$.) A set $\Sigma$ of semisimple $G_0$-orbit representatives in $\g_1$ can now be obtained by taking the union of the sets of $\Gamma_{\Pi_i}$-representatives in $\h_{\Pi_i}^\circ$ for all~$i$.


\subsection{Mixed orbits}\label{sec:mixed}

\noindent We now consider elements of \emph{mixed type}, that is, $x\in \g_1$
that is neither nilpotent nor semisimple. The investigation of such elements
is based on their Jordan decomposition (see Lemma \ref{lem:jd}).

Let $\Sigma$ be the set of semisimple orbit representatives from the previous
section; we also use the sets $\h_{\Pi_i}^\circ$ defined above. If $x=s+n$ is a mixed element in $\g_1$, then  there exists a unique $s'\in\Sigma$ such that $g(s)=s'$ for some
$g\in G_0$, hence $g(x)=s'+g(n)$ with $g(n)$ nilpotent and $[s',g(n)]=0$.
Together with the uniqueness of the Jordan decomposition, this shows the
following.

\begin{lemma}\label{lemSE}
Every element of mixed type in $\g_1$ is $G_0$-conjugate to an element in
\[\mathcal{M}=\{s+n : s\in \Sigma \text{ and nonzero nilpotent } n\in\z_{\g_1}(s)\},\]
and $s+n, s'+n'\in\mathcal{M}$ are $G_0$-conjugate if and only if $s=s'$ and
$n'=g(n)$ for some $g\in Z_{G_0}(s)$.
\end{lemma}
 
Because of this lemma we are interested in determining the centraliser
$Z_{G_0}(s)$ of semisimple $s$. This requires the following preliminary
lemma, which seems to be well-known; we include a proof because we could not
find a precise reference in the literature.

\begin{lemma}\label{lemZG}
Let $K\leq \GL(n,\C)$ be a connected reductive algebraic group with semisimple 
Lie algebra $\mathfrak{k}\leq \gl(n,\C)$. If  $x\in \mathfrak{k}$ is semisimple, then
$Z_K(x) =  \{ g\in K : gxg^{-1}=x\}$ is connected.
\end{lemma}

\begin{proof}
  Let $U$ be the smallest algebraic subgroup of $K$ whose Lie algebra $\u$ contains $x$; such a subgroup always exists and it is unique and connected, see \cite[Theorem 4.1.5]{graaf}. It follows from \cite[Lemma~4.7.3]{graaf} that $Z_K(U)=\{g\in K : gh=hg \text{ for all } u\in U\}$ equals $Z_K(\mathfrak{u})=\{g\in K: gu=ug \text{ for all } u\in\mathfrak{u}\}$. Let $A$ be  the
associative matrix algebra with identity generated by $x$. It follows from  \cite[Example~3.6.9]{graaf} that the unit group $A^*$ of $A$ is an algebraic subgroup
of $\GL(n,\C)$ with Lie algebra $A$ where the Lie bracket is given by
the commutator. By \cite[Theorem 4.1.5]{graaf} we  have  $U\leq A^\ast$. Since elements in $A$ are linear combinations of powers of $x$, it follows that every $w\in Z_K(x)$ also centralises $A^\ast$, in particular, $w\in Z_K(U)$, thus $Z_K(x)\leq Z_K(U)=Z_G(\u)$. Since $x\in\u$, we have $Z_K(\u)\leq Z_K(x)$, hence $Z_K(x)=Z_K(U)$. Since $U\leq A^\ast$ consists of commuting semisimple elements, $U\leq K$ is a subtorus; now it follows from \cite[Corollary 8.13a)]{malletest} that $Z_K(U)$ is connected.
\end{proof}
 
The next lemma shows that the determination of the centralisers for the
infinitely many $p\in\Sigma$ can be reduced to a finite calculation: it
suffices to consider one explicit element in each $\h_{\Pi_i}^\circ$. For $x\in \g_1$ we denote its centraliser in $\g$ by $\z_g(x)$.

\begin{lemma}\label{lemCent}
If $x,y\in \h_{\Pi_i}^\circ$, then $\z_\g(x)=\z_\g(y)$ and $Z_{G_0}(x)=Z_{G_0}(y)$.
\end{lemma}

\begin{proof}
Note that $\z_\g(x)=\z_\g(y)$ as both  equal $\h \oplus \bigoplus\nolimits_{\alpha\in \Pi_i} \g_\alpha$. Intersecting with $\g_0$ yields   $\z_{\g_0}(x) = \z_{\g_0}(y)$.
As in the proof of Lemma \ref{lemZG}, let $U_x,U_y\leq G$ be the minimal
subtori whose Lie algebras $\u_x$ and $\u_y$ contain $x$ and $y$, respectively;
the proof also showed that $Z_G(x)=Z_G(U_x)=Z_G(\u_x)$ and $Z_G(y)=Z_G(U_y)=
Z_G(\u_y)$ are both connected.  Moreover, both groups  have the same Lie
algebra $\z_\g(x)=\z_\g(y)$: consider  the adjoint representation ${\rm Ad}
\colon G \to \GL(\g)$ with differential $\ad \colon \g \to \mathfrak{gl}(\g)$.
It follows from \cite[Corollary~4.2.8]{graaf} that $Z_G(x) =
\{ g\in G : {\rm Ad}(g)(x)=x \}$ has Lie algebra $\{ y\in \g : \ad y (x) = 0\}
= \z_\g(x)$. Now  \cite[Theorem 4.2.2]{graaf}  shows that  $Z_{G}(x)=Z_{G}(y)$,
and so $Z_{G_0}(x)=Z_{G_0}(y)$.
\end{proof}

\begin{remark}Now let $q_1,q_2\in \h_{\Pi_i}^\circ$, so $\z_\g(q_1) = \z_\g(q_2)$ and
$Z_{G_0}(q_1) = Z_{G_0}(q_2)$. Write $\a = \z_\g(q_1)$ and $A_0 = Z_{G_0}(q_1)$.
We note that $\a=\a_0\oplus\a_1$ is graded with each $\a_i = \g_i \cap \a$,  the Lie algebra of $A_0$ is $\a_0$, and $A_0$ acts on $\a_1$. Let $n_1,\ldots, n_s$ be representatives of the nilpotent
$A_0$-orbits in $\a_1$. 
Then the $G_0$-orbits of mixed elements with semisimple parts $q_1$ and $q_2$ have
representatives $q_1+n_i$ and $q_2+n_i$, respectively, for $i\in\{1,\ldots,s\}$. In particular, the nilpotent parts of mixed elements with semisimple part in $\h_{\Pi_i}^\circ$ do not depend on the choice of the particular semisimple element. In order to determine the nilpotent elements $n_i$, we first determine the nilpotent orbits
of the identity component $A_0^\circ$ acting on $\a_1$.  This can be done by an algorithm that only works with the Lie algebra $\a$, see  \cite[Section 8.4]{graaf}. The fusion of these $A_0^\circ$-orbits in $A_0$ can be decided by computing a set of representatives in $A_0$ of the the component group $A_0/A_0^\circ$.
\end{remark}

\section{The orbits of $\SL(2,\C)^4$ on $\mathcal{H}_4$}\label{secOrbit} 

\noindent In this section we classify the orbits of the group
\[\widehat{G}=\SL(2,\C)^4\]
acting on the space $\mathcal{H}_4=\C^2\otimes \C^2\otimes \C^2\otimes \C^2$. First, we show how this action comes
from a symmetric pair as studied in the previous section.

Let  $\g$ be the simple Lie algebra of type D$_4$ defined over the complex
numbers. Let $\Psi$ denote its root system with respect to a fixed Cartan
subalgebra $\mathfrak{t}$. Let $\gamma_1,\ldots,\gamma_4$ be a fixed choice of simple roots such that the Dynkin diagram of $\Psi$ is labelled as follows
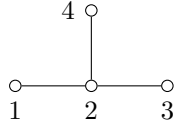
\begin{center}
\begin{tikzpicture}
\node[dnode, label=below:{\small $1$}] (1) at (0,1) {};
\node[dnode, label=below:{\small $2$}] (2) at (1,1) {};
\node[dnode, label=below:{\small $3$}] (3) at (2,1) {};
\node[dnode, label=left:{\small $4$}] (4) at (1,2) {};
\path (1) edge[sedge] (2)
(2) edge[sedge] (3)
(2) edge[sedge] (4);
\end{tikzpicture}
\end{center}
We now construct a $\Z/2\Z$-grading of $\g$: let $\g_0$ be spanned by
$\mathfrak{t}$ along with the root spaces $\g_\gamma$, where $\gamma= \sum_i k_i
\gamma_i$ has $k_2$ even, and let $\g_1$ be spanned by those $\g_\gamma$
where $\gamma= \sum_i k_i \gamma_i$ has $k_2$ odd. Let $\gamma_0=\gamma_1+2\gamma_2+\gamma_3+\gamma_4$ be the highest root
of $\Psi$. The  root system of $\g_0$ is $\{\pm \gamma_0, \pm \gamma_1, \pm \gamma_3,\pm\gamma_4\}$, hence
$$\g_0\cong \ssl(2,\C)^4=\ssl(2,\C)\oplus \ssl(2,\C)\oplus\ssl(2,\C)\oplus\ssl(2,\C).$$
Taking $-\gamma_0,\gamma_1,\gamma_3,\gamma_4$ as basis of simple roots of
$\g_0$ we have that $-\gamma_2$ is the highest weight of the $\g_0$-module
$\g_1$, which therefore is isomorphic to $\mathcal{H}_4= \C^2\otimes \C^2\otimes \C^2\otimes \C^2$. We fix a basis $\{e_0,e_1\}$ of $\C^2$ and denote the basis elements of
$\mathcal{H}_4$ by
\[\mye{i_1i_2i_3i_4}= e_{i_1}\otimes e_{i_2}\otimes e_{i_3}\otimes e_{i_4}.\]
Mapping any nonzero root vector in $\g_{-\gamma_2}$ to $\mye{0000}$
extends uniquely to an isomorphism $\g_1 \to \mathcal{H}_4$ of $\ssl(2,\C)^4$-modules. We denote  by $G$ the adjoint group of $\g$, and we write $G_0$ for the 
connected algebraic subgroup of $G$ with  Lie algebra $\ad_\g \g_0\cong
\mathfrak{sl}(2,\C)^4$. 
The isomorphism $\ssl(2,\C)^4 \to \g_0$ lifts to a surjective 
morphism $\pi\colon \widehat{G}\to G_0$  of algebraic groups, which makes
$\g_1$ into a $\wG$-module isomorphic to $\mathcal{H}_4$.

It is well-known (cf.\ \cite{verstraete}), and we have verified by computer, that 
$$u_1=\mye{0000}+\mye{1111},\;\; u_2=\mye{0110}+\mye{1001},\;\;
u_3=\mye{0101}+\mye{1010},\;\; u_4=\mye{0011}+\mye{1100}$$
span a Cartan subspace $\h$ of $\g_1$. This shows that the symmetric pair $(\g_0,\g_1)$ is
of maximal rank.

Let $\Phi$ be the root system corresponding to the Cartan subalgebra $\h$ of
$\g$. (Note that $\h$ is necessarily different from $\mathfrak{t}$ as
$\h\subset\g_1$ and $\mathfrak{t}\subset \g_0$.)
By $W$ we denote its Weyl group. Representing a root $\alpha$ by the 4-tuple
$(\alpha(u_1),\ldots,\alpha(u_4))$, a choice of simple roots is
$\Delta=\{\alpha_1,\ldots,\alpha_4\}$, where
$$  \alpha_1 = (0,-2,0,0),\;\; \alpha_2 = (1,1,1,1),\;\;  \alpha_3 = (0,0,-2,0),\;\; \alpha_4 = (0,0,0,-2);$$
here we use the same enumeration as in the above Dynkin diagram.

\begin{remark}\label{remSym}
Consider the group $\mathrm{Sym}_4$ of all permutations of
$\{1,2,3,4\}$. For $\sigma\in{\rm Sym}_4$ we define the linear map
$\pi_{\sigma}\colon \g_1\to\g_1$ that maps each $\mye{i_1i_2i_3i_4}$ to
$\mye{i_{1^\sigma}i_{2^\sigma}i_{3^\sigma}i_{4^\sigma}}$. Since $\g_1$ generates
$\g$ as a Lie algebra, there is at most one way in which $\pi_\sigma$
extends to an automorphism of $\g$. We have checked by computer that indeed
for all $\sigma \in \mathrm{Sym}_4$ this yields an automorphism of $\g$. The group generated by all these $\pi_\sigma$ fixes $u_1$ and permutes
$\{u_2,u_3,u_4\}$ as ${\rm Sym}_3$.  Specifically, $\pi_{(2,3)}$ swaps $u_3$ and
$u_4$, and $\pi_{(2,4)}$ swaps $u_2$ and $u_4$. Since $\wG$ has no elements acting as a $\pi_\sigma$, the space $\mathcal{H}_4$ is acted upon by the split product \[\mathcal{S} =\mathrm{Sym}_4\ltimes \SL(2,4)^4.\]
\end{remark} 

\noindent {\bf Classifications:} We use the techniques described above  to classify the $\widehat{G}$-orbits
on $\mathcal{H}_4$, which are exactly the  $G_0$-orbits in $\g_1$.
Our classification of semisimple elements  is given in Theorem \ref{thmSE} and Table \ref{tabW}; our
classification of mixed elements is given in Theorem \ref{thmME}. We then
consider the group $\mathcal{S}$: A classification of semisimple, nilpotent,
and mixed elements up to $\mathcal{S}$-conjugacy is described in Theorem~\ref{thmSConj}.

\subsection{The orbits of semisimple elements}\label{secOurSS}

We use Dynkin's algorithm to compute all root subsystems of $\Phi$. This yields
11 subsystems, up to $W$-conjugacy. One of these subsystems is of type
$4{\rm A}_1$ and therefore not complete. The others are complete and listed
in the third column of Table \ref{tabW}. By adding the empty set we obtain
11 subsystems $\Pi_1,\ldots,\Pi_{11}$.  Table \ref{tabW} also contains the data that we computed starting from the complete
root subsystems. The fourth column has the description of $\h_{\Pi_i}$ and the
fifth column gives the polynomial conditions that an element of $\h_{\Pi_i}$
has to satisfy to belong to $\h_{\Pi_i}^\circ$; see
Remark \ref{remhpc} below for more details. The sixth column describes
the groups $\Gamma_{\Pi_i}$, see also the statement of Theorem \ref{thmSE} for
more details. Here we write $I_4=\diag(1,1,1,1)$ for the $4\times 4$ identity
matrix. Finally, the last column displays the semisimple part of the
reductive centraliser $\z_\g(p_i)$, where $p_i$ is some element in $\h_{\Pi_i}^\circ$, cf.\ Lemma \ref{lemCent}.
 
 \begin{theorem}\label{thmSE}
Up to $G_0$-conjugacy, the semisimple elements of $\g_1$ are the
$\Gamma_{\Pi_i}$-classes of elements in $\h_{\Pi_i}^\circ$ for $i=1,\ldots,11$,
as given in Table \ref{tabW}; each $\Gamma_{\Pi_i}$ is  realised
as a complement subgroup to $W_{\Pi_i}$ in~$N_W(W_{\Pi_i})$: the group $\Gamma_{\Pi_2}\cong (\mathbb{Z}/2\mathbb{Z})^3$ is generated by all $4\times 4$ diagonal matrices that have two $1$s and two $-1$s on the diagonal; the groups $\Gamma_{\Pi_4},\Gamma_{\Pi_5},\Gamma_{\Pi_6}\cong \Dih_4$ are isomorphic to the dihedral group of order $8$ and defined as
\[\Gamma_{\Pi_4}=\langle \SmallMatrix{1&0&0&0\\0&1&0&0\\0&0&1&0\\0&0&0&-1}, \SmallMatrix{0&0&0&-1\\0&0&1&0\\0&1&0&0\\1&0&0&0}\rangle,\;
\Gamma_{\Pi_5}=\langle  \SmallMatrix{1&0&0&0\\0&1&0&0\\0&0&-1&0\\0&0&0&1},\SmallMatrix{0&0&-1&0\\0&0&0&-1\\1&0&0&0\\0&-1&0&0}\rangle,\;
\Gamma_{\Pi_6}=\langle  \SmallMatrix{-1&0&0&0\\0&1&0&0\\0&0&1&0\\0&0&0&1},\SmallMatrix{0&-1&0&0\\1&0&0&0\\0&0&0&-1\\0&0&-1&0}\rangle.
\]
 \end{theorem}

In the following we denote by $\Sigma$ a set of $G_0$-orbit representatives of semisimple elements in $\g_1$.

\begin{table}[ht]
\scalebox{0.87}{\begin{tabular}{r|r|r|c|c|c|c}
  {\bf $\pmb{i}$} & {\bf type of $\Pi_i$} & {\bf roots of $\Pi_i$} & {\bf elements of $\pmb{\h_{\Pi_i}}$} & {\bf condition for being in $\pmb{\h_{\Pi_i}^\circ}$}& $\pmb{\Gamma_{\Pi_i}}$ & $\pmb{\z_\g(p_i)'}$\\
  \hline
  1& $\emptyset$     &  & $\lambda_1u_1+\cdots+\lambda_4u_4$ &
  $\lambda_i\neq 0$ and  $\lambda_1\notin \{\pm \lambda_2\pm \lambda_3\pm \lambda_4\}$ & $W$ & $0$\\
  2 & ${\rm A}_1$ & $\alpha_4$ & $\lambda_1u_1+\lambda_2u_2+\lambda_3u_3$ &
  $\lambda_i\neq 0$ and  $\lambda_1\notin\{ \pm\lambda_2\pm\lambda_3\}$& $(\Z/2\Z)^3$ & $\ssl(2,\C)$\\
  3 & ${\rm A}_2$ & $\alpha_2,\alpha_4$ & $\lambda_1 (u_1-u_2) + \lambda_2 (u_1-u_3)$ &
   $\lambda_u\neq 0$ and $\lambda_1\ne -\lambda_2$& $\langle-I_4\rangle$ & $\ssl(3,\C)$\\
  4 & $2{\rm A}_1$ & $\alpha_1,\alpha_3$ & $\lambda_1u_1+\lambda_2u_4$ &
  $\lambda_i\neq 0$ and $\lambda_1\notin \{\pm\lambda_2\}$& $\Dih_4$ & $\ssl(2,\C)^2$\\
  5 & $2{\rm A}_1$ & $\alpha_1,\alpha_4$ & $\lambda_1u_1+\lambda_2u_3$ &
  $\lambda_i\neq 0$ and $\lambda_1\notin \{\pm\lambda_2\}$& $\Dih_4$ & $\ssl(2,\C)^2$\\
  6 & $2{\rm A}_1$ & $\alpha_3,\alpha_4$ & $\lambda_1u_1+\lambda_2u_2$ &
    $\lambda_i\neq 0$ and $\lambda_1\notin \{\pm\lambda_2\}$& $\Dih_4$ & $\ssl(2,\C)^2$\\
  7 & ${\rm A}_3$ & $\alpha_1,\alpha_2,\alpha_3$ & $\lambda_1(u_1-u_4)$ & $\lambda_1\neq 0$& $\langle-I_4\rangle$ & $\ssl(4,\C)$\\
  8 & ${\rm A}_3$ & $\alpha_1,\alpha_2,\alpha_4$ & $\lambda_1(u_1-u_3)$ & $\lambda_1\neq 0$& $\langle-I_4\rangle$ & $\ssl(4,\C)$\\
  9 & ${\rm A}_3$ & $\alpha_2,\alpha_3,\alpha_4$ & $\lambda_1(u_1-u_2)$ & $\lambda_1\neq 0$& $\langle-I_4\rangle$ & $\ssl(4,\C)$\\
  10 & $3{\rm A}_1$ & $\alpha_1,\alpha_3,\alpha_4$ & $\lambda_1 u_1$ & $\lambda_1\neq 0$& $\langle-I_4\rangle$ & $\ssl(2,\C)^3$\\
  11 & ${\rm D}_4$ & $\alpha_1,\ldots,\alpha_4$  & $0$ & $0$& $1$ & $\so(4,\C)$
\end{tabular}}\\[1ex]\caption{Complete root subsystems $\Pi_i$ of $\Phi$, corresponding 
sets $\h_{\Pi_i}$ and $\h_{\Pi_i}^\circ$ with parameters $\lambda_1,\ldots,\lambda_4\in\C$, and groups $\Gamma_{\Pi_i}$; the last column displays
the derived algebra of the centraliser $\z_\g(p_i)$ for $p_i \in
\h_{\Pi_i}^\circ$.}\label{tabW} 
\end{table}

\begin{remark}\label{GammaConjugacy}
  To determine $\Sigma$ explicitly,  one still has to consider the  $\Gamma_{\Pi_i}$-classes of elements in $\h_{\Pi_i}^\circ$ for every $i=1,\ldots,10$; note that the case $i=11$ only contributes the zero element.  In view of Remark \ref{remSym}, the cases $i=4,5,6$ are all symmetric and classifications for one of these $i$ can be translated to classifications for the other two cases by applying one of the automorphisms $\pi_{(2,3)}$ or $\pi_{(2,4)}$ of $\g$; similarly, the cases $i=7,8,9$ are symmetric. If $i\in\{3,7,8,9,10\}$, then the elements listed under $\h_{\Pi_i}$ in column four of Table \ref{tabW}  need to be reduced modulo multiplication by $-1$. If $i=2$, then  we can change two signs of $u_1,u_2,u_3,u_4$ at once; since $u_4$ is not involved in elements of $\h_{\Pi_2}$, for fixed $\lambda_1,\lambda_2,\lambda_3$, all the elements $e_1\lambda_1u_1+e_2\lambda_2u_2+e_3\lambda_3u_3\in \h_{\Pi_2}^\circ$ with $e_1,e_2,e_3\in\{\pm1\}$ are $\Gamma_{\Pi_2}$-conjugate. If $i=4$, then we can swap $u_1$ and $u_4$, or change their signs; cases $i\in\{5,6\}$  are analogous. For $i=1$ we act with $W$; a direct calculation shows that every element in $W$ acts as  $PQ^i$ for some $i\in\{0,1,2\}$ where
  \[Q=\tfrac12\SmallMatrix{ 1& -1& -1& 1\\ 1& -1& 1& -1 \\  -1& -1& 1& 1 \\ 1& 1& 1& 1}\]and $P$ is any signed permutation matrix induced by an element in  $\langle (1,2)(3,4), (1,3)(2,4)\rangle$.  
\end{remark}

\begin{remark}\label{remhpc} We comment on the construction of  $\h_{\Pi_i}^\circ$ as listed in Table  \ref{tabW}. Recall that  $p\in \h_{\Pi_i}$ lies in $\h_{\Pi_i}^\circ$ if and only if $\beta(p)\neq 0$ for all $\beta\in\Phi\setminus\Pi_i$.
Representing
a positive root by the 4-tuple of its values on $u_1,\ldots,u_4$, we have $\alpha_1 = (0,-2,0,0)$, $\alpha_2 = (1,1,1,1)$, $\alpha_3 = (0,0,-2,0)$, $\alpha_4 = (0,0,0,-2)$, and 
\begin{align*}
  \alpha_1+\alpha_2&=(1,-1,1,1),& \alpha_2+\alpha_3&=(1,1,-1,1)\\
  \alpha_2+\alpha_4 &= (1,1,1,-1), &\alpha_1+\alpha_2+\alpha_3 &= (1,-1,-1,1), \\
  \alpha_1+\alpha_2+\alpha_4 &=   (1,-1,1,-1),  &\alpha_2+\alpha_3+\alpha_4&=(1,1,-1,-1),\\
  \alpha_1+\alpha_2+\alpha_3+\alpha_4&=(1,-1,-1,-1),
  &\alpha_1+2\alpha_2+\alpha_3+\alpha_4&=(2,0,0,0).
\end{align*}
Consider the third line of Table \ref{tabW}, so $p= \lambda_1(u_1-u_2)+\lambda_2(u_1-u_3)$; the coordinate vector with respect to the chosen basis of $\h$ is $(\lambda_1+\lambda_2,-\lambda_1,-\lambda_2,0)$. The positive roots in $\Pi_3$ are
$\alpha_2,\alpha_4,\alpha_2+\alpha_4$. For $p$ to be in $\h_{\Pi_3}^\circ$, the inner product of
$(\lambda_1+\lambda_2,-\lambda_1,-\lambda_2,0)$ with all positive roots other
than $\{\alpha_2,\alpha_4,\alpha_2+\alpha_4\}$ has to be nonzero; it is straightforward to see that this
reduces to the condition $\lambda_1\lambda_2(\lambda_1+\lambda_2)\neq 0$. In a similar way, one can determine the conditions for the other sets $\h_{\Pi_i}^\circ$.
\end{remark}

\subsection{The orbits of mixed elements}\label{secMixed}

As before, $\widehat{G}=\SL(2,\C)^4$; recall that there  is a surjective morphism $\pi\colon \widehat{G}\to G_0$  of
algebraic groups and $\widehat{G}$ acts as $G_0$ on $\g_1$. We start with an observation.

\begin{lemma}\label{lemCentSL}
If  $x,y\in \h_{\Pi_i}^\circ$ for some $i$, then $Z_{\widehat{G}}(x) = Z_{\widehat{G}}(y)$.
\end{lemma}

\begin{proof}
By Lemma \ref{lemCent} we have $Z_{G_0}(x)=Z_{G_0}(y)$. Note that
$Z_{\widehat{G}}(x)$ is the preimage of $Z_{G_0}(x)$ in $\widehat{G}$ under
$\pi$, and similarly for $Z_{\widehat{G}}(y)$; therefore
$Z_{\widehat{G}}(x)=Z_{\widehat{G}}(y)$. 
\end{proof}

\begin{proposition}\label{propCent}
  If $s\in\Sigma$ lies in  $\h_{\Pi_i}^\circ$ as in Table \ref{tabW}, then $Z_{\widehat{G}}(s)$ is given in Row $i$ of Table \ref{tabZ}.
\end{proposition}
\begin{proof}
  The group $\widehat{G}$ acts naturally on $\g_1\cong \C^2\otimes \C^2\otimes  \C^2\otimes  \C^2$, and we can write down the equations for $g(s)=s$ where $g=(A,B,C,D)\in\SL(2,\C)^4$ is a general element with $16$ indeterminates $a_{ij},b_{i,j},c_{ij},d_{ij}$ defining the matrices $A,B,C,D$. We use Gr\"obner basis techniques to obtain a useful description of $Z_{\widehat{G}}(s)$. Table \ref{tabZ} summarises our results, where for $A=\SmallMatrix{a&b\\c&d}$ and $u,v\in\C$ with $u\ne 0$ we write  
\begin{equation}\label{eqmats}
  \begin{array}{rclrclrclrlcrclrclr} 
    A^\#\!\!\!\!&=&\!\!\!\!\SmallMatrix{d&c\\b&a},  & D(u,v)\!\!\!\!&=&\!\!\!\!\SmallMatrix{u&0\\v&u^{-1}}, &D(u)\!\!\!\!&=&\!\!\!\!\SmallMatrix{u&0\\0&u^{-1}}, & L(v)\!\!\!\!&=&\!\!\!\!D(1,v), & L\!\!\!\!&=&\!\!\!\!D(\imath)\\[2ex]
    M(a,b)\!\!\!\!&=&\!\!\!\!\SmallMatrix{a&b\\b&a},&I\!\!\!\!&=&\!\!\!\!\SmallMatrix{1&0\\0&1}, & J\!\!\!\!&=&\!\!\!\!\SmallMatrix{0&1\\-1&0}, & K\!\!\!\!&=&\!\!\!\!\SmallMatrix{0&\imath \\\imath &0}.
  \end{array}\qedhere
\end{equation}
\end{proof}

In Table \ref{tabZ} we often have $(J,J,J,J)^2=-(I,I,I,I)\in Z_{\widehat{G}}(s)^\circ$, in which case $Z_{\widehat{G}}(s)^\circ$ does not split in $Z_{\widehat{G}}(s)$, that is, the generators listed in the right column of Table \ref{tabZ} generate a group that is larger than the component group $Z_{\widehat{G}}(s)/Z_{\widehat{G}}(s)^\circ$.

\begin{table}\renewcommand\arraystretch{1.4} 
 \scalebox{0.9}{\begin{tabular}{ccc}
    {\bf $\pmb{i}$}  & {\bf identity component $\pmb{ Z_{\widehat{G}}(s)^\circ}$} & {\bf preimages of generators of $\pmb{Z_{\widehat{G}}(s)/Z_{\widehat{G}}(s)^\circ}$} \\\hline
	 1 & 1 & $(J,J,J,J),(-I,-I,I,I),(-I,I,-I,I),(K,K,K,K)$ \\\hline
	 
	 2 & $\left\{ (D(a)^{-1}, D(a)^{-1}, D(a), D(a)) \;:\; a\in \C^\times \right\}$ &
$ (-I,-I,I,I),(-I,I,-I,I),(J,J,J,J)$         \\\hline
	 
	 3 & $\left\{(A^\#, A^\#, A, A) \;:\;
	 A \in \SL(2,\C)\right\}$ & $(-I,-I,I,I),(-I,I,-I,I)$ \\\hline
	 
	 4 & $\left\{ (D(a)^{-1}, D(a), D(b)^{-1}, D(b)) \;:\; a,b\in \C^\times\right \}$ & $(-I,I,-I,I),(J,J,J,J)$ \\\hline
         5 & $\left\{ (D(a)^{-1}, D(b)^{-1}, D(a), D(b)) \;:\; a,b\in \C^\times\right \}$ & $(-I,-I,I,I),(J,J,J,J)$ \\\hline
         6 & $\left\{ (D(a)^{-1}, D(b), D(b)^{-1}, D(a)) \;:\; a,b\in \C^\times\right \}$ & $(-I,I,-I,I),(J,J,J,J)$ \\\hline

	 7 & $\left\{ (A^\#, A, B^\#, B) : A,B\in \SL(2,\C)\right\}$ & $(-I,I,-I,I)$ \\\hline
	 8 & $\left\{ (A^\#, B^\#, A, B) : A,B\in \SL(2,\C)\right\}$ & $(-I,-I,I,I)$ \\\hline
         9 & $\left\{ (A^\#, B, B^\#, A) : A,B\in \SL(2,\C)\right\}$ & $(-I,I,-I,I)$ \\\hline
         
	 10 & $\left\{ (D(abc)^{-1}, D(a), D(b), D(c)) \;:\; a,b,c\in \C^\times \right\}$ &  $(J,J,J,J)$ \\[2ex]
  \end{tabular}}\caption{The groups $Z_{\widehat{G}}(s)$: the entry $i$ is the label of the canonical semisimple set $\h_{\Pi_i}^\circ$ that contains $s$, as in Table \ref{tabW}; the notation $A^\#$,  $D(a)$, $I$, $J$, and $K$ is as explained in \eqref{eqmats}.}\label{tabZ}   
\end{table}

By Lemma \ref{lemSE}, up to $G_0$-conjugacy, every mixed element has the form  $p+e$ where
$p\in\Sigma$ is semisimple (as in Table \ref{tabW}) and $e\in \z_\g(p)' \cap \g_1$ is nilpotent; recall that $\z_\g(p)$ is reductive and its center consists of semisimple elements, so $e$ lies in the semisimple part $\z_\g(p)'$. Moreover, $p+e$ and $p+e'$ are $G_0$-conjugate if and only if $e$ and $e'$ are
$Z_{G_0}(p)$-conjugate. Writing $\a = \z_\g(p)'$ and $\a_i = \a \cap \g_i$ for $i=0,1$, we need to classify the nilpotent $Z_{G_0}(p)$-orbits in $\a_1$.

Lemma \ref{lemCent} shows that any element in the same row of Table \ref{tabW} has the same centraliser. Moreover, the proof of Lemma \ref{lemCent} shows that $Z_G(p)$ is connected with Lie algebra $\z_\g(p)$. It follows from $p\in\g_1$ that $\z_\g(p)=\z_{\g_0}(p)\oplus \z_{\g_1}(p)$; thus, $\z_\g(p)$ is a reductive graded Lie algebra with adjoint group $Z_G(p)$, and $Z_{G_0}(p)^\circ\leq Z_G(p)$ is the connected algebraic subgroup with Lie algebra $\z_{\g_0}(p)$.  We can now use  standard methods, such as described in \cite[Chapter 8.3.2]{graaf}, to classify the (finitely many) nilpotent $Z_{G_0}(p)^\circ$-orbits in $\a_1$; we have done this in GAP \cite{gap} using the GAP package SLA. It remains to reduce the obtained list up to conjugacy under the component group of $Z_{G_0}(p)$.

According to Table \ref{tabW}, we have 11 cases, namely $p\in \Sigma\cap\h_{\Pi_i}^\circ$ with $i\in\{1,\ldots,11\}$. For $i=1$ we have $\a=0$, so there are no nonzero nilpotent elements, and $p=0$ for $i=11$; in both cases there are no mixed elements. Thus, it remains to consider $i\in \{2,\ldots,10\}$; we report on the outcome of our computations:

\begin{iprf}
\item[{\bf Case $i=2$.}]
Here $\a=\ssl(2,\C)$ and there are two nilpotent
$Z_{\widehat{G}}(p)^\circ$-orbits in $\a_1$; these are interchanged by the
component group. In conclusion, one nilpotent orbit remains with representative
\[n_{2,1}=\mye{0011}.\]

\item[{\bf Case $i=3$.}]Here $\a=\ssl(3,\C)$ and its grading is induced by an outer automorphism of $\a$. There are two nilpotent $Z_{\widehat{G}}(p)^\circ$-orbits in $\a_1$ and the component group acts trivially. In conclusion, there are two nilpotent orbits with  representatives \[n_{3,1}=\mye{0011}\quad\text{and}\quad n_{3,2}=\mye{0111}+\mye{1011}+\mye{0010}+\mye{0001}.\]

\item[{\bf Case $i=4,5,6$.}]
First let $i=4$. Here $\a=\ssl(2,\C)^2$ and  there are eight nilpotent
$Z_{\widehat{G}}(p)^\circ$-orbits in $\a_1$. Up to the action of the
component group, four of them remain, with  representatives 
\begin{align*}n_{4,1}&=\mye{0110}+\mye{1010},& n_{4,2}&=\mye{0110}+\mye{0101},& n_{4,3}&=\mye{0110},& n_{4,4}&=\mye{0101}.
\end{align*}
For $i=5,6$, Remark \ref{remSym} yields $n_{5,1}=\mye{0110}+\mye{1100}$, $n_{5,2}=\mye{0110}+\mye{0011}$, $n_{5,3}=\mye{0110}$, $n_{5,4}=\mye{0011}$ and $n_{6,1}=\mye{0011}+\mye{1010}$, $n_{6,2}=\mye{0011}+\mye{0101}$, $n_{6,3}=\mye{0011}$, $n_{6,4}=\mye{0101}$, respectively.
\smallskip

\item[{\bf Case $i=7,8,9$.}]
First let $i=7$. Here $\a=\ssl(4,\C)$ and its grading is
induced by an outer automorphism of $\a$;  we have that $\a_0=\ssl(2,\C)^2$ and  there are six nilpotent
$Z_{\widehat{G}}(p)^\circ$-orbits in $\a_1$. The component group acts trivially, and representatives of nilpotent orbits are 
\begin{align*}n_{7,1}&=\mye{1101}+\mye{1011}+\mye{1000}+\mye{0001},&n_{7,2}&=\mye{1101}+\mye{1010}+\mye{0001},\\ n_{7,3}&=\mye{1011}+
\mye{1000}+\mye{0101},& n_{7,4}&=\mye{1011}+\mye{1000},\\ n_{7,5}&=\mye{1101}+\mye{0001},& n_{7,6}&=\mye{1001}.
\end{align*}
Remark \ref{remSym} yields $n_{8,1}=\mye{1011}+\mye{1101}+\mye{1000}+\mye{0001}$, $n_{8,2}=\mye{1011}+\mye{1100}+\mye{0001}$, $n_{8,3}=\mye{1101}+\mye{1000}+\mye{0011}$, $n_{8,4}=\mye{1101}+\mye{1000}$, $n_{8,5}=\mye{1011}+\mye{0001}$, $n_{8,6}=\mye{1001}$ for the case $i=8$, and $n_{9,1}=\mye{1101}+\mye{1110}+\mye{1000}+\mye{0100}$, $n_{9,2}=\mye{1101}+\mye{1010}+\mye{0100}$, $n_{9,3}=\mye{1110}+\mye{1000}+\mye{0101}$, $n_{9,4}=\mye{1110}+\mye{1000}$,  $n_{9,5}=\mye{1101}+\mye{0100}$, $n_{9,6}=\mye{1100}$ for $i=9$.
\smallskip

\item[{\bf Case $i=10$.}]
Here $\a=\ssl(2,\C)^3$ and  there are 26 nilpotent $Z_{\widehat{G}}(p)^\circ$-orbits in $\a_1$. Up to the action of the
component group, 13 of them remain, with representatives
\begin{align*}
  n_{10,1}&=\mye{1100}+\mye{1010}+\mye{0110}, &n_{10,2}&=\mye{1010}+\mye{0110}, &  n_{10,3}&=\mye{1010}+\mye{0110}+\mye{0011},\\
  n_{10,4}&=\mye{1100}+\mye{0110},& n_{10,5}&=\mye{0110}, & n_{10,6}&=\mye{0110}+\mye{0011},\\
  n_{10,7}&=\mye{1100}+\mye{0110}+\mye{0101}, & n_{10,8}&=\mye{0110}+\mye{0101},& n_{10,9}&=\mye{0110}+\mye{0101}+\mye{0011},\\
  n_{10,10}&=\mye{1100}+\mye{1010},& n_{10,11}&=\mye{1010},& n_{10,12}&=\mye{1010}+\mye{0011},\\
  n_{10,13}&=\mye{0011}.
\end{align*}
\end{iprf}

\noindent In conclusion, we have shown:

\begin{theorem}\label{thmME}
For $i\in\{2,\ldots,10\}$ let $\Sigma_i$ be a set of $G_0$-conjugacy representatives of semisimple elements in $\h_{\Pi_i}^\circ$ as specified in Table \ref{tabW} and Remark \ref{GammaConjugacy}. Up to $G_0$-conjugacy, the mixed type elements in $\g_1$ are the elements $s+n_{i,j}$ where $i\in\{2,\ldots,10\}$, $s\in \Sigma_i$, and $n_{i,j}$ as specified in Case $i$ above.
\end{theorem}


\section{A classification up to $\mathcal{S}$-conjugacy}\label{secDjok}
\noindent This section has two aims. First, we compare our classifications with the families determined by Verstraete et al.\ \cite{verstraete}; the latter families have been reconsidered and corrected by  Chterental \& Djokovi\v{c} \cite{djok}. Second, we describe our classification of nilpotent, semisimple, and mixed elements up to $\mathcal{S}$-conjugacy where $\mathcal{S}= {\rm Sym_4}\ltimes \SL_2(\C)^4$. We start with a preliminary section on deciding conjugacy of elements.


\subsection{Deciding conjugacy}\label{sec:conj}
\noindent Let $u,v\in \mathcal{H}_4$. In our classification below we need to decide whether $u$ and $v$  lie in the same
$\widehat{G}$-orbit and, if so, to find a $g\in \widehat{G}$ with $gv=u$.
A general method for this is based on the computational technique
of Gr\"obner bases \cite{CLO15}: the relation $gv=u$ gives linear
relations on the entries of the four matrices in $g$; to these relations we add the polynomials
that express that the determinants of the matrices are 1. We then compute a Gr\"obner basis of the ideal generated by the resulting
polynomials. This Gr\"obner basis is trivial (i.e., consists only of 1)
if and only if there is no solution. If the Gr\"obner basis is not trivial,
then in many cases it can be used to effectively solve the equations and find
a solution.
A related problem is to find, given a $u\in \mathcal{H}_4$, an element $v$
in our classification to which $u$ is conjugate to.

First, suppose $u$ is semisimple.  It follows from \cite[Theorem 3]{vinberg} that two semisimple 
$u$ and $v$ are conjugate if and only if $\mathcal{F}(u) = \mathcal{F}(v)$, where $\mathcal{F}$ is defined in \eqref{eqF} below (it maps $u$ to the values of the generating invariants of $\mathbb{C}[\g_1]^{\wG}$). We compute
$\mathcal{F}(u)$ with Table \ref{tabSSinv}, and use  a Gr\"obner basis
computation to find an element $v$ in one of the 10 semisimple classes with
$\mathcal{F}(u) = \mathcal{F}(v)$; we then find a conjugating element via the above method. If $u$ is nilpotent, then we have to perform at most 30 Gr\"obner basis
computations to find the element in Table \ref{tabNP} that is conjugate to
$u$. By computations in $\g$, we can also reduce the number of candidates
of nilpotent elements in Table \ref{tabNP} that are possibly conjugate
to $u$. For example, conjugate elements have centralisers in $\g_0$ of the same dimension. Moreover, the theory
of $\ssl_2$-triples can be used to reduce the number of candidates, see \cite[\S 8.3.2]{graaf}. If $u$ is mixed, then we first identify its semisimple part
with an element in our classification;  subsequently we deal with the
nilpotent part.

\subsection{Classification results}
Before we describe our classification, we first recall the classification in \cite{djok} in the language of our paper:

\begin{theorem}[Theorem 3.6 in \cite{djok}]\label{thmDjok}
The $\mathcal{S}$-orbits on $\g_1$ are classified by the nine families $D_1$,\ldots,$D_9$ in Table~\ref{tabDjok}.  Elements belonging to different families are not equivalent under $\mathcal{S}$-operations. However, within the same family, different families of the parameters may give elements belonging to the same $\mathcal{S}$-orbit.
\end{theorem}

\renewcommand\arraystretch{1.2} 
\begin{table}[ht]\scalebox{0.85}{
    \begin{tabular}{rl}
 {\bf fam.} & {\bf elements}\\\hline     
  $D_1$ & $S_1(a,b,c,d)+N_1$ where  $S_1(a,b,c,d)=\tfrac{a+d}{2} u_1 + \tfrac{b-c}{2} u_2 + \tfrac{b+c}{2}u_3 + \tfrac{a-d}{2} u_4$ and $N_1=0$\\\hline
 
  $D_2$ & $S_2(a,b,c)+N_2$ where  $S_2(a,b,c)=\tfrac{a+c}{2} u_1 + \tfrac{b-c}{2} u_2 + \tfrac{b+c}{2}u_3 + \tfrac{a-c}{2} u_4$ and \\
  & $N_2=\tfrac{\imath}{2}\left(u_3+u_4-u_2-u_1+  \mye{1110}+\mye{0001}+\mye{1000}+\mye{0111}-\mye{1101}-\mye{0010}-\mye{1011}-\mye{0100} \right)$\\\hline

  $D_3$ & $S_3(a,b)+N_3$ where  $S_3(a,b)=\tfrac{a}{2} u_1 + \tfrac{b}{2}u_2 + \tfrac{b}{2}u_3 + \tfrac{a}{2}u_4$ and \\ 
  & $N_3=\tfrac{1}{2}\left( u_3-u_2+\mye{0010}+\mye{1101}-\mye{1110}-\mye{0001}\right)$\\\hline

  $D_4$ & $S_4(a,b)+N_4$ where $S_4(a,b)=\tfrac{a+b}{2}u_1+ bu_3 +\tfrac{a-b}{2} u_4$ and\\
  & $N_4=\imath(\mye{1001}-\mye{0110})+\tfrac{1}{2}\left(\mye{1101}+\mye{0100}+\mye{1011}+\mye{0010}-\mye{1110}-\mye{0001}-\mye{1000}-\mye{0111}\right)$\\\hline

  $D_5$ &  $S_5(a)+N_5$ where $S_5(a)=au_1 + au_3 $ and  $N_5=2\imath\left(\mye{0001}+\mye{0110}-\mye{1011}\right)$\\\hline
  $D_6$ & $S_6(a)+N_6$ where $S_6(a)=\tfrac{a}{2} u_1 +\tfrac{a}{2}u_2+\tfrac{a}{2}u_3+\tfrac{a}{2}u_4$ and \\
  & $N_6=\tfrac{\imath+1}{2}(\mye{0010}+\mye{1101}-u_2)+\tfrac{\imath-1}{2}(\mye{1110}+\mye{0001}-u_3)-\tfrac{\imath}{2}(\mye{1011}+\mye{0100}+\mye{1000}+\mye{0111}-u_1-u_4)$\\\hline

  $D_7$ & $S_7+N_7$ where  $S_7=0$ and\\ &$N_7=(\mye{1010}-\mye{1001}+\mye{0011}+\mye{0000})+(\imath+1)(\mye{0110}+\mye{0101})-\imath(\mye{1011}+\mye{1000}+\mye{0010}-\mye{0001})$\\\hline
  $D_8$ & $S_8+N_8$ where  $S_8=0$ and\\ &$N_8=\tfrac{\imath+1}{2}u_1-\tfrac{\imath-1}{2}u_4+\tfrac{\imath-1}{2}(\mye{1110}+\mye{0001})-\tfrac{\imath+1}{2}(\mye{1101}+\mye{0010})$\\&{\color{white}$N_8=\;$}$+\tfrac{1}{2}(\mye{1011}+\mye{0110}+\mye{0101}+\mye{1000})+\tfrac{1-2\imath}{2}(\mye{0111}+\mye{1010}+\mye{1001}+\mye{0100})$\\\hline
  $D_9$ & $S_9+N_9$ where  $S_9=0$ and \\ & $N_9=\tfrac{1}{2}(\mye{1111}+\mye{1100}+\mye{1011}+\mye{1000}+\imath\mye{1110}+\imath\mye{1101}-\imath\mye{1010}+\imath\mye{1001})$
   \\&{\color{white}$N_9=\;$}$+\tfrac{1}{2}(\mye{0111}+\mye{0100}+\mye{0011}+\mye{0000}+\imath\mye{0110}+\imath\mye{0101}-\imath\mye{0010}+\imath\mye{0001})$
\end{tabular}
}\\[1ex]\caption{The nine families of Theorem \ref{thmDjok} with parameters $a,b,c,d\in \C$.}\label{tabDjok}
\end{table}

The next theorem identifies in which of these nine families our $G_0$-orbit representatives lie, up to $\mathcal{S}$-conjugacy; we also present a new, complete and irredundant classification up to $\mathcal{S}$-conjugacy.
\pagebreak
\begin{theorem}\label{thmSConj}
 \begin{ithm} 
 \item Up to $\mathcal{S}$-conjugacy, the nilpotent orbits in $\g_1$ are the elements $N_1,\ldots,N_9$ in Table \ref{tabDjok}.
 \item Up to $\mathcal{S}$-conjugacy, the semisimple orbits in $\g_1$ are the elements in Table \ref{tabSSS}.
 \item Up to $\mathcal{S}$-conjugacy, the mixed  elements in $\g_1$ are the elements  in Table \ref{tabMT}.
 \end{ithm}\noindent The right column in Table \ref{tabSSS} and the second column in Table \ref{tabMT}  indicate to which family $D_i$ (as in Table \ref{tabDjok}) the element is $\mathcal{S}$-conjugate to. Tables \ref{tabZ}, \ref{tabMT}, and \ref{tabNilZ} contain information about the centralisers in $\wG$.
\end{theorem}
\begin{proof}
  As before,  all the direct computations mentioned in this proof have been carried out in GAP \cite{gap} and its interface to Singular \cite{singular}. We briefly comment on our approach; let $s,t\in \g_1$ and let $\sigma\in{\rm Sym}_4$. By abuse of notation, we also denote by $\sigma$ the induced automorphism of $\g$, see Remark \ref{remSym}. It is straightforward to compute the image $\sigma(s)$.  As in the proof of Proposition~\ref{propCent}, we use Gr\"obner basis techniques to determine $\wG$-conjugacy of $\sigma(s)$ and $t$: for example, if $g=(A,B,C,D)\in\wG$ is a general element with $16$ indeterminates $a_{ij},b_{i,j},c_{ij},d_{ij}$, then the command {\small{\tt HasTrivialGroebnerBasis}} allows us to decide quickly whether a solution to $g(\sigma(s))=t$ exists. This approach can also be used if $s$ and $t$ are semisimple or mixed elements defined by parameters $\lambda_i$ and $\lambda_i'$: if the Gr\"obner basis is trivial, then the elements are not conjugate; if nontrivial, then the elements are \emph{potentially conjugate}. In the latter situation one still has to determine whether a solution exists that satisfies the conditions on the parameters $\lambda_i$ and $\lambda_i'$. As explained below, we usually reduce $\wG$-conjugacy testing to  testing of $W$-conjugacy, see Proposition~\ref{prop:gampsi}; the latter is a finite explicit calculation.

 \begin{iprf}
\item Table \ref{tabDjok} yields nine elements $N_1,\ldots,N_9$, and a direct calculation shows that they are all nilpotent. Another direct computation (using Gr\"obner bases) shows that all these elements are not $\mathcal{S}$-conjugate, as expected by Theorem \ref{thmDjok}. It has been determined in \cite{nilp} that there are 31 nilpotent $G_0$-orbits in $\g_1$; the corresponding classification over the reals has been presented in \cite{nilporb}. In Table \ref{tabNP} (left) we list representatives for the nilpotent orbits (taken from \cite[Table~I]{nilporb}) and determine (using Gr\"obner bases) to which nilpotent element $N_1,\ldots, N_9$ the element is $\mathcal{S}$-conjugate to; the claim follows. 

\item Recall that $u_1,\ldots,u_4$ span the Cartan subspace $\h$, which shows that all the elements $S_1,\ldots,S_9$ in Table \ref{tabDjok} are semisimple. By Theorem \ref{thmSE}, every semisimple element is conjugate to an element in family~$D_1$. It remains to reduce our classification of semisimple elements (as given in Table \ref{tabW}) up to $\mathcal{S}$-conjugacy. Due to Remark \ref{remSym}, it suffices to consider elements in $\h_{\Pi_i}^\circ$ for $i\in\{1,2,3,4,7,10\}$. First, we note that if $s\in\h_{\Pi_i}^\circ$ and $t\in\h_{\Pi_j}^\circ$ with distinct $i,j\in \{1,2,3,4,7,10\}$,  then $s$ and $t$ are not $\mathcal{S}$-conjugate: this follows because $s$ and $t$ have centralisers of different dimensions (see Table \ref{tabW}) and because the permutation action of every $\sigma\in{\rm Sym}_4$ on $\g_1$ extends to Lie algebra automorphisms of $\g$ (cf.\ Remark \ref{remSym}). Thus, it remains to determine when $s,t\in\h_{\Pi_i}^\circ$ are $\mathcal{S}$-conjugate;  note that $G_0$-conjugacy is already determined in Table \ref{tabW}. We also note that  ${\rm Sym}_4$ stabilises  $\h$, so $G_0$-conjugacy of ${\rm Sym}_4$-conjugate elements in $\h$ can be decided by considering the action of $W$, see Proposition \ref{prop:gampsi}.  For example, consider $i=3$. Elements $s=\lambda_1(u_1-u_2)+\lambda_2(u_1-u_3)$ and $t=\lambda_1'(u_1-u_2)+\lambda_2'(u_1-u_3)$ are $G_0$-conjugate if and only if $(\lambda_1',\lambda_2')=\pm(\lambda_1,\lambda_2)$. We now consider every  ${\rm Sym}_4$-conjugate of $s$, for example $s'=\lambda_1(u_1-u_4)+\lambda_2(u_1-u_2)$, and then determine $\lambda_1'$ and $\lambda_2'$ such that $s'$ is $W$-conjugate to $\lambda_1'(u_1-u_2)+\lambda_2'(u_1-u_3)$. In this particular example, $(\lambda_1',\lambda_2')=(\lambda_1+\lambda_2,-\lambda_1)$, which shows that $\lambda_1(u_1-u_2)+\lambda_2(u_1-u_3)$ is $\mathcal{S}$-conjugate to $(\lambda_1+\lambda_2)(u_1-u_2)-\lambda_1(u_1-u_3)$.   Doing this for all ${\rm Sym}_4$-conjugates of semisimple representatives in $\h_{\Pi_i}^\circ$ allows us to determine the conditions for $\mathcal{S}$-conjugacy; the result is listed in Table \ref{tabSSS}.
    
\item Let $x=s+n_{i,j}$ be a mixed element as in Theorem \ref{thmME}. Due to Remark \ref{remSym}, up to $\mathcal{S}$-conjugacy, we can assume that $i\in\{2,3,4,7,10\}$. (There are no mixed elements for $i=1$.)  As in part a), we first determine to which nilpotent element $N_1,\ldots,N_9$ the element $n_{i,j}$ is $\mathcal{S}$-conjugate to; this is listed in Table~\ref{tabNP} (right). If we have determined that $n_{i,j}$ is $\mathcal{S}$-conjugate to $N_k$, then we use Gr\"obner basis computations to verify that $x$ is indeed $\mathcal{S}$-conjugate to an element of the form $S_k+N_k$, hence, up to $\mathcal{S}$-conjugacy, $x$ lies in family $D_k$. Recall that ${\rm Sym}_4$ preserves $\h$ and the action of every $\sigma\in{\rm Sym}_4$ on $\g_1$ extends to a Lie algebra automorphism of $\g$. In particular, it follows that $\sigma(x)=\sigma(s)+\sigma(n_{i,j})$ is the Jordan decomposition of $\sigma(x)$. The centraliser information in Table \ref{tabW} now implies that the only $\mathcal{S}$-conjugacies between elements of the form $s+n_{i,j}$ (with  $i\in\{2,3,4,7,10\}$, $s\in\h_{\Pi_i}^\circ\setminus\{0\}$, and permissible~$j$) are between elements whose semisimple parts lie in the same component $\h_{\Pi_i}^\circ$. Thus, it remains to decide $\mathcal{S}$-conjugacy of elements $x=s+n_{i,j}$ and $y=t+n_{i,\ell}$ where $i\in \{2,3,4,7,10\}$, $s,t\in \h_{\Pi_i}^\circ$ as in Table \ref{tabW}, and $n_{i,j}$ and $n_{i,\ell}$ in the same family $D_k$. In this situation, explicit Gr\"obner basis computations show that if $s=t$, then  $x$ and $y$ are $\mathcal{S}$-conjugate. It therefore remains to consider $\mathcal{S}$-conjugacy between elements \[x=s+n_{i,j}\quad\text{and}\quad y=t+n_{i,j}\] with $s,t\in\h_{\Pi_i}^\circ$ as in Table \ref{tabW}; by what is said in the previous sentence, for each $i$ we only have to consider one $j$ for each class $D_k$.  For this we consider every possible ${\rm Sym}_4$-conjugate of $x$, say $x'=\sigma(x)=\sigma(s)+\sigma(n_{i,j})$, and check whether $x'$ is potentially $G_0$-conjugate to an element $t+n_{i,j}$ with $t\in\h_{\Pi_i}^\circ$: for this we first check whether $\sigma(n_{i,j})$ is potentially conjugate to $n_{i,j}$, and if so, then we test the same for $\sigma(x)$ and $y$. If the test is positive, then we check whether $\sigma(s)$ is $W$-conjugate to an element $t$, cf.\ the proof of part b). We briefly comment on each case.

  If $i=7$, then $s=\lambda_1 u_1$ and $t=\lambda_1'u_1$, and a direct computation shows that $x$ and $y$ are $\mathcal{S}$-conjugate if and only if $s=\pm t$; the same holds for $i=10$. If $i=4$, then $s=\lambda_1u_1+\lambda_2 u_4$ and $t=\lambda_1'u_1+\lambda_2' u_4$. A computation shows that $\sigma(x)$ and $y$ are potentially $G_0$-conjugate  if and only if $\sigma(s)=s$. More precisely, if $j=3$, then $\sigma(x)=s+\sigma(n_{4,3})$, and the latter can be shown to be $G_0$-conjugate to $s+n_{4,3}$. If $j=1$, then  $\sigma(x)=s+\sigma(n_{4,1})$ is $G_0$-conjugate to $s+n_{4,1}$.  In conclusion, for $i=4$ it follows that $x$ and $y$ are $\mathcal{S}$-conjugate if and only if $s$ and $t$ are $G_0$-conjugate. If $i=2$, then for every permutation $\sigma$, the element $\sigma(x)=\sigma(s)+\sigma(n_{2,1})$ has mixed type and must be $G_0$-conjugate to some $t+n_{2,1}$; this follows from the centraliser dimension in Table \ref{tabW} and Theorem \ref{thmME}. We can determine the possible transformations $s\to\sigma(s)\to t$ by using the same computations as in b).  Now consider $i=3$. One can show that every ${\rm Sym}_4$-conjugate of $s$ is $G_0$-conjugate to an element in $\h_{\Pi_3}^\circ$ as in Table \ref{tabW}; since the possible nilpotent parts $n_{3,1}$ and $n_{3,2}$ lie in different families $D_k$, they are not $G_0$-conjugate, thus, if  $x=s+n_{3,i}$, then each $\sigma(x)$ is  $G_0$-conjugate to an element $t+n_{3,i}$ with $t\in \h_{\Pi_3}^\circ$ as in Table \ref{tabW}, and we determine the possible transformations $s\to\sigma(s)\to t$ as in b). All the results are listed in Table \ref{tabMT}.
 \end{iprf}
\end{proof}

\renewcommand\arraystretch{1.2} 
\begin{table}\scalebox{0.86}{
    \begin{tabular}[ht]{rclr}
      \multicolumn{1}{c}{\bf element} & \multicolumn{1}{c}{\bf component}& \multicolumn{1}{c}{\bf conditions} & \multicolumn{1}{c}{\bf family}\\\hline

      $\lambda_1u_1+\lambda_2u_2+\lambda_3u_3+\lambda_4u_4$ & $\h_{\Pi_1}^\circ$ &with  $\lambda_1,\ldots,\lambda_4\ne 0$ and $\lambda_1 \notin \{\pm\lambda_2\pm\lambda_3\pm \lambda_4\}$ & $D_1$ \\&& up to the action of $PQ^i$ with $i\in\{0,1,2\}$ where  $Q$ as \\&& in Remark \ref{GammaConjugacy}d) and $P$ is any $4\times 4$ signed permutation matrix \\\hline

    $\lambda_1u_1+\lambda_2u_2+\lambda_3u_3$ & $\h_{\Pi_2}^\circ$ & with  $\lambda_1,\lambda_2,\lambda_3\ne 0$ and $\lambda_1 \notin \{\pm\lambda_2\pm\lambda_3\}$ & $D_1$ \\&& up to the action of $3\times 3$ signed permutation matrices\\\hline
            
$\lambda_1(u_1-u_2)+\lambda_2(u_1-u_3)$  & $\h_{\Pi_3}^\circ$ &with  $\lambda_1,\lambda_2\ne 0$ and $\lambda_1 \ne -\lambda_2$\\&& up to action of  $\langle \left(\begin{smallmatrix}0&1\\1&0\end{smallmatrix}\right),\left(\begin{smallmatrix} 1&1\\0&-1\end{smallmatrix}\right)\rangle\cong {\rm Dih}_6$ & $D_1$\\ \hline

   $\lambda_1u_1+\lambda_2u_4$  & $\h_{\Pi_4}^\circ$ &with  $\lambda_1,\lambda_2\ne 0$ and $\lambda_1 \ne \pm\lambda_2$\\&&  up to the action of  $\langle \left(\begin{smallmatrix} -1&0\\0&1\end{smallmatrix}\right),\left(\begin{smallmatrix}0&1\\1&0\end{smallmatrix}\right)\cong{\rm Dih}_4$& $D_1$\\\hline

$\lambda_1(u_1-u_4)$ & $\h_{\Pi_7}^\circ$& with $\lambda_1\ne 0$ up to the action of $\langle (-1)\rangle$ & $D_1$\\\hline

  $\lambda_1 u_1$ & $\h_{\Pi_{10}}^\circ$& with $\lambda_1\ne 0$ up to the action of $\langle (-1)\rangle$ & $D_1$\\\hline
                
\end{tabular}}\\[1ex]\caption{The classification of semisimple elements up to $\mathcal{S}$-conjugacy; the action of each matrix group is on the vector of parameters $(\lambda_1$), $(\lambda_1,\lambda_2)$, etc. The corresponding centraliser of an element in $\h_{\Pi_i}^\circ$ is given in Row $i$ of Table \ref{tabZ}. }\label{tabSSS}
\end{table}

\renewcommand\arraystretch{1.4} 
\begin{table}\scalebox{0.78}{
	\hspace*{0cm}\begin{tabular}[H]{cccc}
		{\bf element}  & {\bf fam.} & {\bf identity component $\pmb{ Z^\circ}$} & {\bf preimages of generators of $\pmb{Z/Z^\circ}$} \\\hline
		
		$s+n_{2,1}$ & $D_2$ & 1 & $ (-I,-I,I,I),(-I,I,-I,I),(L,L,L,L)$ \\\hline
		
		$s+n_{3,1}$ & $D_2$ & $\left\{ (L(a')^\intercal, L(a')^\intercal, L(a'), L(a')) \;:\; a'\in \C \right\}$ & $ (-I,-I,I,I),(-I,I,-I,I),(L,L,L,L)$  \\\hdashline
		
		$s+n_{3,2}$ & $D_4$ & 1 & $ (-I,-I,I,I),(-I,I,-I,I),(-I,-I,-I,-I)$  \\\hline
		
		$s+n_{4,1}$ & $D_3$ & 1 & $ (-I,-I,I,I),(-I,I,-I,I),(L,L,L,L)$   \\\hdashline
		
		$s+n_{4,3}$ & $D_2$ & $\left\{ (D(a)^{-1}, D(a), D(a)^{-1}, D(a)) \;:\; a\in \C^\times \right\}$ & $ (-I,-I,I,I),(-I,I,-I,I)$  \\\hline
		
		$s+n_{7,1}$ & $D_4$ & $\left\{ (L(a')^{-1}, L(a')^{-\intercal}, L(a')^\intercal, L(a')) \;:\; a'\in \C \right\}$ & $ (-I,-I,I,I),(-I,I,-I,I),(-I,-I,-I,-I)$  \\\hdashline
		
		$s+n_{7,2}$ & $D_5$ & 1 & $ (-I,-I,I,I),(-I,I,-I,I),(-I,-I,-I,-I)$  \\\hdashline
		
		$s+n_{7,4}$ & $D_3$ & $\left\{ (L(a'),L(a')^\intercal, D(b)^{-1}, D(b)) \;:\; a'\in \C, b\in \C^\times \right\}$ & $ (-I,-I,I,I),(-I,I,-I,I), (L,L,J,J)$  \\\hdashline
		
		$s+n_{7,6}$ & $D_2$ &  $(D(a^{-1},c'), D(a,c')^\intercal, D(a^{-1},b')^{\intercal}, D(a,b'))$ & $ (-I,-I,I,I),(-I,I,-I,I)$  \\ & & with  $a\in \C^\times$ and $b',c'\in\C$ \\\hline
		
		$s+n_{10,1}$ & $D_6$ & 1 & $ (-I,-I,I,I),(-I,I,-I,I),(L,L,L,L)$ \\\hdashline
		
		$s+n_{10,2}$ & $D_3$ & $\left\{ (D(a)^{-1}, D(a)^{-1}, D(a), D(a)) \;:\; a\in \C^\times \right\}$ & $ (-I,-I,I,I),(-I,I,-I,I)$  \\\hdashline
		
		$s+n_{10,5}$ & $D_2$ & $\left\{ (D(a)^{-1}, D(b)^{-1}, D(b), D(a)) \;:\; a,b\in \C^\times\right \}$ & $(-I,-I,I,I)$ \\\hdashline
		
\end{tabular}}\\[1ex]\caption{The classification of mixed elements up to $\mathcal{S}$-conjugacy; for each listed element $s+n_{i,j}$ the semisimple part $s\in \h_{\Pi_i}^\circ$ is as given in Table \ref{tabSSS}. Last two columns describe their centralisers, where the notation is from \eqref{eqmats}.}\label{tabMT} 
\end{table}

 

\renewcommand\arraystretch{1} 
\begin{table}[ht]\scalebox{0.88}{
\begin{minipage}[t]{8cm}\hspace*{-1.7cm}\begin{tabular}{rlc}
{\bf orbit}  & {\bf representative} & $\pmb{\mathcal{S}}${\bf-conjugate to} \\\hline
1 & $\mye{1100}$ & $N_2$ in $D_2$\\
2 & $\mye{1100}+\mye{0000}$ & $N_3$ in $D_3$\\
3 & $\mye{1100}+\mye{1001}$ & $N_3$ in $D_3$\\
4 & $\mye{1100}+\mye{1010}$ & $N_3$ in $D_3$\\
5 & $\mye{1101}+\mye{0100}$ & $N_3$ in $D_3$\\
6 & $\mye{1110}+\mye{0100}$ & $N_3$ in $D_3$\\
7 & $\mye{1110}+\mye{1101}$ & $N_3$ in $D_3$\\
8 & $\mye{1101}+\mye{0100}+\mye{1000}$ & $N_6$ in $D_6$\\
9 & $\mye{1110}+\mye{0100}+\mye{1000}$ & $N_6$ in $D_6$\\
10 & $\mye{1110}+\mye{1101}+\mye{1000}$ & $N_6$ in $D_6$\\
11 & $\mye{1110}+\mye{1101}+\mye{0100}$ & $N_6$ in $D_6$\\
12 & $\mye{0101}+\mye{1100}+\mye{1001}+\mye{0000}$ & $N_9$ in $D_9$\\
13 & $\mye{0110}+\mye{1100}+\mye{1010}+\mye{0000}$ & $N_9$ in $D_9$\\
14 & $\mye{1111}+\mye{1100}+\mye{1001}+\mye{1010}$ & $N_9$ in $D_9$\\
15 & $\mye{0111}+\mye{1110}+\mye{1101}+\mye{0100}$ & $N_9$ in $D_9$\\
16 & $\mye{1110}+\mye{1101}+\mye{0100}+\mye{1000}$ & $N_4$ in $D_4$\\
17 & $\mye{1110}+\mye{1101}+\mye{0000}$ & $N_5$ in $D_5$\\
18 & $\mye{1110}+\mye{0100}+\mye{1001}$ & $N_5$ in $D_5$\\
19 & $\mye{1101}+\mye{0100}+\mye{1010}$ & $N_5$ in $D_5$\\
20 & $\mye{0101}+\mye{1110}+\mye{1000}$ & $N_5$ in $D_5$\\
21 & $\mye{0110}+\mye{1101}+\mye{1000}$ & $N_5$ in $D_5$\\
22 & $\mye{1111}+\mye{0100}+\mye{1000}$ & $N_5$ in $D_5$\\
23 & $\mye{1110}+\mye{0100}+\mye{0000}+\mye{1001}$ & $N_8$ in $D_8$\\
24 & $\mye{0110}+\mye{1101}+\mye{1000}+\mye{0000}$ & $N_8$ in $D_8$\\
25 & $\mye{1111}+\mye{0100}+\mye{1000}+\mye{1001}$ & $N_8$ in $D_8$\\
26 & $\mye{1111}+\mye{0100}+\mye{1000}+\mye{1010}$ & $N_8$ in  $D_8$\\
27 & $\mye{0101}+\mye{1110}+\mye{0000}+\mye{1001}$ & $N_7$ in $D_7$\\
28 & $\mye{0110}+\mye{1101}+\mye{0000}+\mye{1010}$ & $N_7$ in $D_7$\\
29 & $\mye{1111}+\mye{0100}+\mye{1001}+\mye{1010}$ & $N_7$ in $D_7$\\
30 & $\mye{1111}+\mye{0110}+\mye{0101}+\mye{1000}$ & $N_7$ in $D_7$\\
31 & 0
\end{tabular}
\end{minipage}\hspace*{1cm}

\begin{minipage}[t]{5cm}\vspace*{-7.4cm}\begin{tabular}{cc}
{\bf element} $\pmb{n_{i,j}}$  & $\pmb{\mathcal{S}}${\bf-conjugate to} \\\hline
$n_{2,1}$ & $N_2$ in $D_2$\\
$n_{3,1}$ & $N_2$ in $D_2$\\
$n_{3,2}$ &  $N_4$ in $D_4$\\
$n_{4,1}$ &  $N_3$ in $D_3$\\
$n_{4,2}$ &  $N_3$ in $D_3$\\
$n_{4,3}$ &  $N_2$ in $D_2$\\
$n_{4,4}$ &  $N_2$ in $D_2$\\
$n_{7,1}$ &  $N_4$ in $D_4 $\\
$n_{7,2}$ &  $N_5$ in $D_5 $\\
$n_{7,3}$ &  $N_5$ in $D_5  $\\
$n_{7,4}$ &  $N_3$ in $D_3 $\\
$n_{7,5}$ &  $N_3$ in $D_3  $\\
$n_{7,6}$ &  $N_2$ in $D_2 $\\
$n_{10,1}$ &  $N_6$ in $D_6$\\
$n_{10,2}$ &  $N_3$ in $D_3$\\
$n_{10,3}$ &  $N_6$ in $D_6  $\\
$n_{10,4}$ &  $N_3$ in $D_3$\\
$n_{10,5}$ &  $N_2$ in $D_2$\\
$n_{10,6}$ &  $N_3$ in $D_3$\\
$n_{10,7}$ &  $N_6$ in $D_6  $\\
$n_{10,8}$ &  $N_3$ in $D_3$\\
$n_{10,9}$ &  $N_6$ in $D_6$\\
$n_{10,10}$ &  $N_3$ in $D_3$\\
$n_{10,11}$ &  $N_2$ in $D_2$\\
$n_{10,12}$ &  $N_3$ in $D_3$\\
$n_{10,13}$ &  $N_2$ in $D_2$\\
\end{tabular}
\end{minipage}
}\\[1ex]\caption{Complex nilpotent orbits (left) and nilpotent elements $n_{i,j}$ from  Theorem \ref{thmME} (right).}\label{tabNP}
\end{table}

\renewcommand\arraystretch{1.4} 
\begin{table}\scalebox{0.77}{
	\hspace*{0cm}\begin{tabular}[H]{cccc}
		{\bf fam.} & {\bf $\pmb{i}$}  & {\bf identity component $\pmb{ Z^\circ}$} & {\bf preimages of generators of $\pmb{Z/Z^\circ}$} \\\hline
		N2 & 1 & $\{(D(b^{-1}cd,a'), D(b,b'), D(c,c')^\intercal, D(d,d')^\intercal)$ &  \\
                   &   &  with $a',b',c',d'\in\C$ and $b,c,d\in\C^\times$      &    \\\hline
		
		N3 & 2 &  $\left\{ (B^{-\intercal}, B, D(d^{-1},c')^\intercal, D(d,d')^\intercal) \;:\; B\in \SL(2,\mathbb C), d\in \C^\times,c',d'\in\C \right\}$ & $(-I,I,-I,I)$  \\\hline
		
		N4 & 16 &  {$\left\{ (L(b'+c'+d')^{-1}, L(b'), L(c')^\intercal, L(d')^\intercal) \;:\; b',c',d'\in \C \right\}$} & $ (-I,-I,I,I),(-I,I,-I,I), (-L,L,L,L)$   \\\hline
		
		N5 & 17 &  $\left\{ (D(b)^{-1}, D(b), L(d')^{-\intercal}, L(d')^\intercal) \;:\; b\in \C^\times,d\in\C \right\}$ & $ (-I,I,-I,I),(-I,I,I,-I)$  \\\hline	
		
		N6 & 8 &  $(D(d^{-1},-(b'+d')), D(d^{-1},b'), D(d^{-1},c')^\intercal, D(d,d')^\intercal) $  & $ (-I,-I,I,I),(-I,I,-I,I)$  \\
                   &   &  with $d\in\C^\times$ and $b',c',d'\in\C$  \\\hline
		N7 & 27 & 1 & $ (-I,-I,I,I),(-I,I,-I,I),(-I,I,I,-I)$  \\\hline
		
		N8 & 23 &  $\left\{ (L(a'),I, L(a')^\intercal, L(a')^\intercal) \;:\; a'\in \C \right\}$ & $ (-I,-I,I,I),(-I,I,-I,I), (-I,I,I,-I)$  \\\hline
		
		N9 & 12 &  {$( M(c,d)^{-1}M(a,b)^{-1}, M(c,d), L(u)^\intercal, M(a,b))$}  &  $(-I,I,-I,I), (L,L,L,L)$ \\  & &
                with  $a,b,c,d,u\in \C$ and $a^2=1+b^2$ and $c^2=1+d^2$\\ \hline
\end{tabular}}\\[1ex]\caption{The centralisers $Z=Z_{\wG}(e)$ where $e$ is the representative of the nilpotent orbit labelled $i$ in  Table \ref{tabNP}; the notation is explained in \eqref{eqmats}.}\label{tabNilZ}  
\end{table}

 \section{Invariants}\label{secApps}
 \noindent The aim of this section is to describe the invariant ring $R=\mathbb{C}[\g_1]^{\widehat{G}}$.  Let $B=\{b_1,\ldots,b_{16}\}$ be the  basis of $\g_1$ such that
$b_1=\mye{1111}$, $b_2=\mye{1110}$, $b_3=\mye{1101}$, $b_4=\mye{1100}$,
\ldots in lexicographical ordering. Let $\C[\g_1]$ be the ring of polynomial
functions on $\g_1$. We identify $\C[\g_1]$ with the polynomial ring
$\C[x_1,\ldots,x_{16}]$ using the basis $B$, so $f\in
\C[x_1,\ldots,x_{16}]$ is identified with the polynomial function on $\g_1$
that maps  $p\in \g_1$ to $f(c_1,\ldots,c_{16})$ where $c_1,\ldots,c_{16}$ are
the coefficients of $p$ with respect to $B$.

The group $\widehat{G}$ acts on $\C[\g_1]$ by $g\cdot f (x) = f(g^{-1}\cdot x)$, and the invariant ring $\C[\g_1]^{\widehat{G}}$ consists of all polynomials
$f\in \C[\g_1]$ such that $g\cdot f =f$ for all $g\in\widehat{G}$. The
invariants are interesting in our context because they are invariant on
orbits. By a celebrated theorem of Hilbert, $\C[\g_1]$ is finitely generated.
Vinberg has proved
a generalization of Chevalley's restriction theorem,  see \cite[Theorem 7]{vinberg} or \cite[Theorem~3.62]{wallach}, showing that the
restriction map $\C[\g_1]^{\widehat{G}} \to \C[\h]^W$ is an isomorphism.
Moreover, the degrees of the generating invariants of $\C[\h]^W$ are known to be
$2,4,4,6$ (this can be read from Table \ref{tabMaxSp}
by setting $n=4$ in the fourth row and recalling the isomorphism
$\mathrm{SO}(4,\C)=\SL(2,\C)^2$). It follows that
$\C[\g_1]^{\widehat{G}}$ is generated by four homogeneous invariants of
degrees $2,4,4,6$. Formulas for generating invariants have been determined  by Luque \& Thibon \cite{inv}; in this reference they are denoted $H$, $L$, $M$, $D_{xt}$.
In Table~\ref{tabInv} we give their explicit form, where we write $D$
instead of $D_{xt}$. We have checked the
correctness of these expressions by computer in the following way:
Let $\C[\g_1]_k$ denote the space of homogeneous polynomials of degree $k$.
There is a canonical isomorphism of $\widehat{G}$-modules $\C[\g_1]_k \to
\mathrm{Sym}^k (\g_1^\ast)$, where $\g_1^\ast$ denotes the dual module of
$\g_1$. Under this isomorphism, every invariant spans a trivial 1-dimensional
submodule. For $k=2,4,4,6$ we determined the trivial 1-dimensional submodules
of $\mathrm{Sym}^k (\g_1^\ast)$ by linear algebra methods using the Lie algebra
of $\widehat{G}$; this way we found the same invariants as Luque \& Thibon. We now define \begin{eqnarray}\label{eqF}\mathcal{F}\colon \g_1\to \mathbb{C}^4,\quad \mathcal{F}(s)=(H(s),L(s),M(s),D(s)).
\end{eqnarray}For a 4-tuple $v\in\mathbb{C}^4$ denote by $U_v=\{s\in\g_1 : \mathcal{F}(s)=v\}$ the corresponding fibre of $\mathcal{F}$; all these fibres partition $\g_1$. Recall that $e\in \g_1$ is nilpotent if and only if there are $g_1,g_2,\ldots \in \wG$ with $\lim_{i\to \infty}g_i(e)=0$, see \cite[Proposition 1]{vinberg}. Since $\mathcal{F}$ is polynomial, this implies that $\mathcal{F}(0)=\lim_{i\to\infty} \mathcal{F}(g_i(e))=\mathcal{F}(e)$. In particular, if $p+e\in U_v$ is a mixed element, so $[p,e]=0$, then we can assume that each $g_i\in Z_{\wG}(p)$, and hence $\mathcal{F}(p+e)=\mathcal{F}(g_i(p+e))=\mathcal{F}(p+g_i(e))$, with limit $\mathcal{F}(p)=v$. By \cite[Theorem 3]{vinberg}, each fibre $U_v$
consists of a single semisimple orbit $\wG p$ (with $p\in \g_1$ semisimple such that $\mathcal{F}(p)=v$) along with the mixed orbits that have an element in $\wG p$ as their semisimple part; in particular, each fiber is the union of finitely many orbits, cf.\ \cite[Theorem 4]{vinberg}. The different values $\mathcal{F}(s)$ with $s\in\h_{\Pi_i}^\circ$ are listed in Table \ref{tabSSinv}. Furthermore, in Table \ref{tab:invrels} we list generators of the ideal of the polynomial
relations between these values. 

\enlargethispage{0.5cm}
 
\newcommand{\myminus}{$\text{-}$} 

\begin{table} \footnotesize 
\begin{tabular}{p{0.22cm}p{15.5cm}} 
{\bf pol.}\ & \hspace*{2ex}{\bf list of monomials} \\\hline      
$H$ & {\myminus}8.9,\;7.10,\;6.11,\;{\myminus}5.12,\;4.13,\;{\myminus}3.14,\;{\myminus}2.15,\;1.16\\

$L$ & 4.7.10.13,\;{\myminus}4.7.9.14,\;{\myminus}4.6.11.13,\;4.6.9.15,\;4.5.11.14,\;{\myminus}4.5.10.15,\;{\myminus}3.8.10.13,\;3.8.9.14,\;3.6.12.13,\;{\myminus}3.6.9.16,\;{\myminus}3.5.12.14,\;3.5.10.16,\;\linebreak 2.8.11.13,\;{\myminus}2.8.9.15,\;{\myminus}2.7.12.13,\;2.7.9.16,\;2.5.12.15,\;{\myminus}2.5.11.16,\;{\myminus}1.8.11.14,\;1.8.10.15,\;1.7.12.14,\;{\myminus}1.7.10.16,\;{\myminus}1.6.12.15,\;1.6.11.16\\

$M$ & {\myminus}6.7.10.11,\;6.7.9.12,\;5.8.10.11,\;{\myminus}5.8.9.12,\;4.6.11.13,\;{\myminus}4.6.9.15,\;{\myminus}4.5.11.14,\;4.5.9.16,\;{\myminus}3.6.12.13,\;3.6.10.15,\;3.5.12.14,\;{\myminus}3.5.10.16,\;\linebreak {\myminus}2.8.11.13,\;2.8.9.15,\;2.7.11.14,\;{\myminus}2.7.9.16,\;{\myminus}2.3.14.15,\;2.3.13.16,\;1.8.12.13,\;{\myminus}1.8.10.15,\;{\myminus}17.12.14,\;1.7.10.16,\;1.4.14.15,\;{\myminus}1.4.13.16\\

$D$ & {\myminus}4.6.8.9.11.13,\;4.6.8.9.9.15,\;{\myminus}4.6.7.10.11.13,\;4.6.7.9.12.13,\;4.6.7.9.11.14,\;{\myminus}4.6.7.9.9.16,\;4.6.6.11.11.13,\;{\myminus}4.6.6.9.11.15,\;4.5.8.10.11.13,\;\linebreak {\myminus}4.5.8.9.10.15,\;{\myminus}4.5.7.9.12.14,\;4.5.7.9.10.16,\;{\myminus}4.5.6.11.12.13,\;{\myminus}4.5.6.11.11.14,\;4.5.6.10.11.15,\;4.5.6.9.11.16,\;4.5.5.11.12.14,\;\linebreak {\myminus}4.5.5.10.11.16,\;4.4.6.11.13.13,\;{\myminus}4.4.6.9.13.15,\;{\myminus}4.4.5.11.13.14,\;4.4.5.9.14.15,\;3.6.8.10.11.13,\;{\myminus}3.6.8.9.10.15,\;{\myminus}3.6.7.9.12.14,\;\linebreak3.6.7.9.10.16,\;{\myminus}3.6.6.11.12.13,\;3.6.6.9.12.15,\;{\myminus}3.5.8.10.12.13,\;{\myminus}3.5.8.10.11.14,\;3.5.8.10.10.15,\;3.5.8.9.12.14,\;3.5.7.10.12.14,\;\linebreak {\myminus}3.5.7.10.10.16,\;3.5.6.12.12.13,\;3.5.6.11.12.14,\;{\myminus}3.5.6.10.12.15,\;{\myminus}3.5.6.9.12.16,\;{\myminus}3.5.5.12.12.14,\;3.5.5.10.12.16,\;{\myminus}3.4.6.12.13.13,\;\linebreak{\myminus}3.4.6.11.13.14,\;3.4.6.10.13.15,\;3.4.6.9.13.16,\;3.4.5.12.13.14,\;3.4.5.11.14.14,\;{\myminus}3.4.5.10.14.15,\;{\myminus}3.4.5.9.14.16,\;3.3.6.12.13.14,\;\linebreak{\myminus}3.3.6.10.13.16,\;{\myminus}3.3.5.12.14.14,\;3.3.5.10.14.16,\;2.8.8.9.11.13,\;{\myminus}2.8.8.9.9.15,\;{\myminus}2.7.8.9.12.13,\;{\myminus}2.7.8.9.11.14,\;2.7.8.9.10.15,\;2.7.8.9.9.16,\;\linebreak 2.7.7.9.12.14,\;{\myminus}2.7.7.9.10.16,\;{\myminus}2.6.8.11.11.13,\;2.6.8.9.11.15,\;2.6.7.11.12.13,\;{\myminus}2.6.7.9.12.15,\;2.5.8.11.11.14,\;{\myminus}2.5.8.10.11.15,\;\linebreak2.5.8.9.12.15,\;{\myminus}2.5.8.9.11.16,\;{\myminus}2.5.7.11.12.14,\;2.5.7.10.11.16,\;{\myminus}2.4.8.11.13.13,\;2.4.8.9.13.15,\;2.4.7.11.13.14,\;{\myminus}2.4.7.9.14.15,\;\linebreak{\myminus}2.4.6.11.13.15,\;2.4.6.9.15.15,\;2.4.5.11.13.16,\;{\myminus}2.4.5.9.15.16,\;2.3.8.12.13.13,\;{\myminus}2.3.8.10.13.15,\;2.3.8.9.14.15,\;{\myminus}2.3.8.9.13.16,\;\linebreak{\myminus}2.3.7.12.13.14,\;2.3.7.10.13.16,\;2.3.6.11.13.16,\;{\myminus}2.3.6.9.15.16,\;2.3.5.12.14.15,\;{\myminus}2.3.5.12.13.16,\;{\myminus}2.3.5.11.14.16,\;2.3.5.9.16.16,\;\linebreak2.2.8.11.13.15,\;{\myminus}2.2.8.9.15.15,\;{\myminus}2.2.7.11.13.16,\;2.2.7.9.15.16,\;{\myminus}1.8.8.10.11.13,\;1.8.8.9.10.15,\;1.7.8.10.12.13,\;1.7.8.10.11.14,\;\linebreak{\myminus}1.7.8.10.10.15,\;{\myminus}1.7.8.9.10.16,\;{\myminus}1.7.7.10.12.14,\;1.7.7.10.10.16,\;1.6.8.11.12.13,\;{\myminus}1.6.8.9.12.15,\;{\myminus}1.6.7.12.12.13,\;1.6.7.10.12.15,\;\linebreak {\myminus}1.6.7.10.11.16,\;1.6.7.9.12.16,\;{\myminus}1.5.8.11.12.14,\;1.5.8.10.11.16,\;1.5.7.12.12.14,\;{\myminus}1.5.7.10.12.16,\;1.4.8.11.13.14,\;{\myminus}1.4.8.9.14.15,\;\linebreak{\myminus}1.4.7.11.14.14,\;1.4.7.10.14.15,\;{\myminus}1.4.7.10.13.16,\;1.4.7.9.14.16,\;1.4.6.12.13.15,\;1.4.6.11.14.15,\;{\myminus}1.4.6.11.13.16,\;{\myminus}1.4.6.10.15.15,\;\linebreak{\myminus}1.4.5.12.14.15,\;1.4.5.10.15.16,\;{\myminus}1.3.8.12.13.14,\;1.3.8.10.13.16,\;1.3.7.12.14.14,\;{\myminus}1.3.7.10.14.16,\;{\myminus}1.3.6.12.14.15,\;1.3.6.10.15.16,\;\linebreak1.3.5.12.14.16,\;{\myminus}1.3.5.10.16.16,\;{\myminus}1.2.8.12.13.15,\;{\myminus}1.2.8.11.14.15,\;1.2.8.10.15.15,\;1.2.8.9.15.16,\;1.2.7.12.13.16,\;1.2.7.11.14.16,\;\linebreak{\myminus}1.2.7.10.15.16,\;{\myminus}1.2.7.9.16.16,\;1.1.8.12.14.15,\;{\myminus}1.1.8.10.15.16,\;{\myminus}1.1.7.12.14.16,\;1.1.7.10.16.16
\end{tabular}\caption{Generators of  $\mathbb{C}[\g_1]^{\widehat{G}}$: each generator is the sum of the listed monomials, where $i_1.i_2.i_3\ldots $ stands for $x_{i_1} x_{i_2} x_{i_3}\ldots$, and $-i_1.i_2.i_3\ldots$ represents  $-x_{i_1} x_{i_2} x_{i_3}\ldots$}\label{tabInv}
\end{table}

\begin{table}\footnotesize \renewcommand\arraystretch{1.2}
\scalebox{1.13}{\begin{tabular}{rl}
  $\pmb{i}$ & {\bf invariant values $\pmb{\mathcal{F}(s)}$}\\\hline
1&  $(\lambda_1^2+\lambda_2^2+\lambda_3^2+\lambda_4^2,
  -\lambda_1^2 \lambda_2^2+\lambda_1^2 \lambda_3^2+\lambda_2^2 \lambda_4^2-\lambda_3^2 \lambda_4^2,
   \lambda_1^2 \lambda_2^2-\lambda_1^2 \lambda_4^2-\lambda_2^2 \lambda_3^2+\lambda_3^2 \lambda_4^2,$\\
   & \hspace*{1ex}$\lambda_1^4 \lambda_2^2+\lambda_1^2 \lambda_2^4-\lambda_1^2 \lambda_2^2 \lambda_3^2-\lambda_1^2 \lambda_2^2 \lambda_4^2-\lambda_1^2 \lambda_3^2 \lambda_4^2-\lambda_2^2 \lambda_3^2 \lambda_4^2+\lambda_3^4 \lambda_4^2+\lambda_3^2 \lambda_4^4)$\\
2&  $(\lambda_1^2+\lambda_2^2+\lambda_3^2, -\lambda_1^2 \lambda_2^2+\lambda_1^2 \lambda_3^2, \lambda_1^2 \lambda_2^2-\lambda_2^2 \lambda_3^2, \lambda_1^4 \lambda_2^2+\lambda_1^2 \lambda_2^4-\lambda_1^2 \lambda_2^2 \lambda_3^2)$\\
3&   $(2 \lambda_1^2+2 \lambda_1 \lambda_2+2 \lambda_2^2, -\lambda_1^4-2 \lambda_1^3 \lambda_2+2 \lambda_1 \lambda_2^3+\lambda_2^4, \lambda_1^4+2 \lambda_1^3 \lambda_2, 2 \lambda_1^6+6 \lambda_1^5 \lambda_2+6 \lambda_1^4 \lambda_2^2+2 \lambda_1^3 \lambda_2^3)$\\
4&  $( \lambda_1^2+\lambda_2^2, 0, -\lambda_1^2 \lambda_2^2, 0)$\\ 
5&  $( \lambda_1^2+\lambda_2^2, \lambda_1^2 \lambda_2^2, 0, 0)$ \\
6&  $( \lambda_1^2+\lambda_2^2, -\lambda_1^2 \lambda_2^2, \lambda_1^2 \lambda_2^2, \lambda_1^4 \lambda_2^2+\lambda_1^2 \lambda_2^4)$, \\
7&  $( 2 \lambda_1^2, 0, -\lambda_1^4, 0)$\\ 
8&  $( 2 \lambda_1^2, \lambda_1^4, 0, 0)$\\
9&  $( 2 \lambda_1^2, -\lambda_1^4, \lambda_1^4, 2 \lambda_1^6)$ \\ 
10& $( \lambda_1^2, 0, 0, 0) $
\end{tabular}} \caption{Invariant values $\mathcal{F}(s)$ for $s\in\h_{\Pi_i}^\circ$ with parameters $\lambda_1,\ldots,\lambda_4$ as in Table \ref{tabW}.}\label{tabSSinv}
\end{table}

\begin{table}\footnotesize \renewcommand\arraystretch{1.2}
\begin{tabular}{rl}
  $\pmb{i}$ & {\bf relations}\\\hline
2 & $H^5LMD-H^4L^2M^2-H^4LD^2+H^4MD^2-8H^3L^2MD+8H^3LM^2D+8H^2L^3M^2-8H^2L^2M^3$\\
& $-H^3D^3+8H^2L^2D^2-46H^2LMD^2+8H^2M^2D^2+16HL^3MD+64HL^2M^2D+16HLM^3D-16L^4M^2$\\& $-32L^3M^3-16L^2M^4+36HLD^3-36HMD^3-16L^3D^2-24L^2MD^2+24LM^2D^2+
16M^3D^2+27D^4$\\
3 & $H^3D-2H^2LM-4HLD+4HMD+8L^2M-8LM^2-18D^2$, \\
& $H^4-8H^2L+8H^2M-24HD+16L^2+16LM+16M^2$\\
4 & $D$,\; $L$\\
5 & $D$,\; $M$\\
6 & $L+M$,\; $HM-D$\\
7 &  $D$,\; $L$,\; $H^2+4M$\\
8 &  $D$,\; $M$,\; $H^2-4L$\\
9 & $M^3-\tfrac{1}{4}D^2$,\; $L+M$,\; $HD-4M^2$,\; $HM-D$,\; $H^2-4M$\\
10 & $D$, \; $M$,\; $L$\\
\end{tabular} \caption{Generators of the ideal of the polynomial relations between the evaluated invariants for the $i$-th semisimple family. The elements $H,L,M,D$ correspond to the  first to fourth entries in the coordinate vector given in the corresponding row of Table \ref{tabSSinv}.}\label{tab:invrels}
\end{table}

\enlargethispage{1cm}
\bibliographystyle{abbrv}

\begin{thebibliography}{10}

  \bibitem{ES} C.\ Eltschka and J.\ Siewert. {Quantifying entanglement
resources}, J.\ Phys.\ {A47} (2014) 424005. 

\bibitem{HHHH09} R.\ Horodecki, P.\ Horodecki, M.\ Horodecki,  K.\ Horodecki.
{Quantum entanglement}, Rev.\ Mod.\ Phys.\ {81}(2) (2009) 865 (2009). \texttt{quant-ph/0702225}

\bibitem{Ben} C.H. Bennett, S.\ Popescu, D.\ Rohrlich, J.A.\ Smolin, and A.V.\ Thapliyal. Exact and Asymptotic Measures of Multipartite Pure State
Entanglement. Phys.\ Rev.\ A63 (2001) 012307.  \texttt{quant-ph/9908073}

\bibitem{Dur} W.\ D\"{u}r, G.\ Vidal, and J.I.\ Cirac.  {Three qubits
can be entangled in two inequivalent ways}. Phys.\ Rev.\ {A62} (2000) 062314. \texttt{quant-ph/0005115}

\bibitem{Lo-Popescu} H.-K.\ Lo and  S.\ Popescu. {Concentrating
entanglement by local actions -- beyond mean values}. Phys.\ Rev.\ {A63} (2001) 022301. \texttt{quant-ph/9707038}

\bibitem{Dur2} A.\ Acin, E.\ Jane, W.\ D\"{u}r,  and G.\ Vidal. {Optimal
distillation of a GHZ state}. Phys.\ Rev.\ Lett.\ {85} (2000) 4811. \texttt{quant-ph/0007042}

\bibitem{verstraete}
 F.\ Verstraete, J.\ Dehaene, B.\ De Moor, and H.\ Verschelde. Four qubits can be
entangled in nine different ways. Phys.\ Rev.\ A65 (2002) 051001.

 \bibitem{djok}
 O.\ Chterental and D.\ Djokovi\v{c}.
Normal Forms and Tensor Ranks of Pure States of Four Qubits. In Linear Algebra Research Advances, G.\ D.\ Ling (Ed.), Ch.\ 4, 133-167, Nova Science Publ., NY 2007; see \url{arxiv.org/abs/quant-ph/0612184}

\bibitem{Cao-Wang} Y.\ Cao, A.M.\ Wang. Revised Geometric Measure of
Entanglement. Eur.\ Phys.\ J.\ {D44} (2008) 159.  \texttt{quant-ph/0701099}

\bibitem{nilp}
L.\ Borsten, D.\ Dahanayake, M.J.\ Duff, A.\ Marrani, and W.\ Rubens. Four-qubit entanglement classification from string theory. Phys.\ Rev.\ Lett.\ 105 (2010) 100507, 4.

\bibitem{Buniy} R.V. Buniy and T.W.\ Kephart. An algebraic
classification of entangled states. J.\ Phys.\ A45 (2012) 185304. \texttt{arXiv:1012.2630 [quant-ph]}

\bibitem{Chen-Grassl} L. Chen, D.\ Djokovi\'{c}, M.\ Grassl, and B.\ Zeng. Four-qubit pure states as fermionic states. Phys.\ Rev.\ {A88} (2013) 052309. \texttt{arXiv:1309.0791 [quant-ph]}

\bibitem{GA16} M.\ Gharahi, S.J.\ Akhtarshenas.{Entangled graphs: A
classification of four-qubit entanglement}.  Eur.\ Phys.\ J.\  {D70} (2016) 54. \texttt{arXiv:1003.2762 [quant-ph]}

\bibitem{Lamata:2006b} L.\ Lamata, J.\ Le\'{o}n, D.\ Salgado, and E.\ Solano.
{Inductive Entanglement Classification of Four Qubits under SLOCC}.
Phys.\ Rev.\ {A75} (2007) 022318. \texttt{quant-ph/0610233}

\bibitem{Li:2007c} D.\ Li, X.\ Li, H.\ Huang, and X.\ Li. SLOCC
classification for nine families of four-qubits. Quant.\ Info.\ Comp.\ 9 (2007) 0778. \texttt{arXiv:0712.1876 [quant-ph]}

\bibitem{Wallach:2013} N.R.\ Wallach. {Quantum computing and entanglement for mathematicians}. Written report available at \url{math.ucsd.edu/{~}nwallach/Venice2-port-new3.pdf} (2013).

\bibitem{kora}
B.\ Kostant and S.\ Rallis. Orbits and representations associated with
symmetric spaces. Amer.\ J.\ Math.\ 93 (1971) 753--809.

\bibitem{vinberg}
{\`E}.B.\ Vinberg.
The Weyl group of a graded Lie algebra.
{Izv.\ Akad.\ Nauk SSSR Ser.\ Mat.}, 40 (1976) 488-526.
English translation: Math.\ USSR-Izv.\ 10 (1976) 463-495.
  
\bibitem{VE78}
{\`E}.B.~Vinberg and A.G.~{\`E}la{\v{s}}vili.
\newblock A classification of the three-vectors of nine-dimensional space.
\newblock {Trudy Sem.\ Vektor.\ Tenzor.\ Anal.}, 18:197--233, 1978.
\newblock English translation: Selecta Math. Sov., 7, 63-98, (1988).

\bibitem{vinberg79}
{\`E}.B.\ Vinberg. 
Classification of homogeneous nilpotent elements of a semisimple graded
Lie algebra.
{Trudy Semin.\ Vektor.\ Tensor, Anal.}, 19 (1979) 155-177. 
English translation: Selecta Math.\ Sovietica 6 (1987) 15-35.

\bibitem{graaf} 
  W.A.\ de Graaf.
Computation with linear algebraic groups. CRC Press 2017.

\bibitem{GKW} G.\ Gour, B.\ Kraus, N.R.\ Wallach. {Almost all
multipartite qubit quantum states have trivial stabilizer}. J.\ Math.\ Phys.\
{58} (2017) 092204 (2017).   \texttt{arXiv:1609.01327 [quant-ph]}

\bibitem{gap} 
GAP -- Groups, Algorithms and Programming. Available at \url{gap-system.org}.

\bibitem{singular}
W.\ Decker, G.-M.\ Greuel, G.\ Pfister, and H.\ Sch{\"o}nemann.
{\sc Singular} {4-2-1} --- {A} computer algebra system for polynomial computations.
\url{singular.uni-kl.de} (2021).

  \bibitem{wallach}  
N.R.\ Wallach. Geometric invariant theory. Over the real and complex numbers. Universitext. Springer, Cham, 2017.

\bibitem{humcox}
J.E.\ Humphreys.
{Reflection groups and Coxeter groups}.
Cambridge Studies Adv. Math.\ 29. Cambridge University Press 1990. 
  
\bibitem{graaf11}
W.A.~de Graaf.
\newblock Computing representatives of nilpotent orbits of $\theta$-groups.
\newblock J.\ Symb.\ Comp.\ 46 (2011) 438--458.

\bibitem{erdmann}
K.\ Erdmann and M.\ J.\ Wildon.
Introduction to Lie algebras.
Springer Undergraduate Mathematics Series. Springer, 2006.

\bibitem{hum}
J.E.\ Humphreys.
\newblock {Introduction to {L}ie algebras and representation theory}.
Vol.\ 9 Grad.\ Texts Math., 2nd, Springer NY 1978.

\bibitem{antonyan}
L.V.\ Antonyan. Classification of four-vectors of an eight-dimensional space.
(Russian) Trudy Sem.\ Vektor.\ Tenzor.\ Anal.\ 20 (1981) 144–161. 

\bibitem{antelash}
L.V.\ Antonyan, A.G.\ {\`E}lashvili Classification of spinors of
  dimension sixteen. Trudy Tbiliss.\ Mat.\ Inst.\ Razmadze Akad.\ Nauk Gruzin. SSR
  70 (1982) 5-23.

\bibitem{dyn} 
E.B.\ Dynkin.
\newblock Semisimple subalgebras of semisimple {L}ie algebras.
\newblock Mat.\ Sbornik N.S., 30(72) (1952) 349--462 (3 plates).
\newblock English translation: Amer.\ Math.\ Soc.\ Transl.\ 6 (1957)111--244.

\bibitem{malletest}
  G.\ Malle and D.\ Testerman.
  Linear algebraic groups and finite groups of Lie  type.
Cambridge University Press, 2011.

\bibitem{CLO15}
D.A.\ Cox, J.\ Little, and D.\ O'Shea.
\newblock Ideals, varieties, and algorithms.  An introduction to computational algebraic geometry and commutative
  algebra.
\newblock Undergraduate Texts in Mathematics. Springer, Cham, fourth edition,
2015.

\bibitem{nilporb}
H.\ Dietrich, W.A.\ de Graaf, D.\ Ruggeri, and M.\ Trigiante.
\newblock Nilpotent orbits in real symmetric pairs and stationary black holes.
\newblock Fortschr.\ Phys.\ 65 (2017) 2, 1600118.

\bibitem{inv} J.-G.\ Luque and J.-Y.\  Thibon.
Polynomial invariants of four qubits.
Phys.\ Rev.\ A67 (2003) 042303.
  
  
\end{thebibliography}

\end{document}